\documentclass[a4paper,UKenglish,cleveref,autoref,thm-restate,authorcolumns]{lipics-v2019}

\title{Reachability in Fixed Dimension \\ Vector Addition Systems with States} 

\titlerunning{Reachability in Fixed Dimension VASS} 

\author{Wojciech Czerwi\'nski}{University of Warsaw, Poland}{wczerwin@mimuw.edu.pl}{https://orcid.org/0000-0002-6169-868X}{Supported by the ERC grant LIPA, agreement no. 683080.}
\author{S{\l}awomir Lasota}{University of Warsaw, Poland}{sl@mimuw.edu.pl}{https://orcid.org/0000-0001-8674-4470}
{Supported by the NCN grant 2017/27/B/ST6/02093.}
\author{Ranko Lazi\'c}{University of Warwick, UK}{R.S.Lazic@warwick.ac.uk}{https://orcid.org/0000-0003-3663-5182}{Supported by EPSRC grant EP/P020992/1.}
\author{J\'er\^ome Leroux}{CNRS \& University of Bordeaux,
  France}{jerome.leroux@labri.fr}{}{Supported by ANR grant
  ANR-17-CE40-0028.}
\author{Filip Mazowiecki}{Max Planck Institute for Software Systems, Germany}{filipm@mpi-sws.org}{}{}

\authorrunning{W. Czerwi\'nski, S. Lasota, R. Lazi\'c, J. Leroux and F. Mazowiecki} 

\Copyright{Wojciech Czerwi\'nski, S{\l}awomir Lasota, Ranko Lazi\'c, J\'er\^ome Leroux and Filip Mazowiecki} 



\acknowledgements{We thank Matthias Englert for inspiring conversations.}

\nolinenumbers 

\hideLIPIcs  

\EventEditors{John Q. Open and Joan R. Access}
\EventNoEds{2}
\EventLongTitle{42nd Conference on Very Important Topics (CVIT 2016)}
\EventShortTitle{CVIT 2016}
\EventAcronym{CVIT}
\EventYear{2016}
\EventDate{December 24--27, 2016}
\EventLocation{Little Whinging, United Kingdom}
\EventLogo{}
\SeriesVolume{42}
\ArticleNo{23}

\usepackage{todonotes}
\usepackage{algorithm}
\usepackage[noend]{algpseudocode}
\usepackage{xspace}
\usepackage{pgf}
\usepackage{tikz}
\usetikzlibrary{arrows,automata}

\newcommand{\twocol}[5]{
\begin{minipage}{#1\linewidth} #2 \end{minipage}
#3
\begin{minipage}{#4\linewidth} #5 \end{minipage}
}
\newcommand{\moveup}{\vspace{-1mm}}
\newcommand{\movedown}{\vspace{1mm}}

\newcommand{\true}{\mathbf{true}}
\newcommand{\false}{\mathbf{false}}
\newcommand{\NP}{\text{\sc NP}}

\newcommand{\set}[1]{\{#1\}}
\newcommand{\lcmpar}[1]{\text{\sc Lcm}(#1)}
\newcommand{\bin}[1]{\text{\sc Bin}(#1)}
\newcommand{\binext}[2]{\text{\sc Bin}_{#1}(#2)}

\renewcommand{\O}{\mathcal{O}}
\newcommand{\N}{\mathbb{N}}
\newcommand{\Z}{\mathbb{Z}}

\newcommand{\weakexp}[3]{\mathcal{W}_{#1}}

\newcommand{\vv}{\mathbf{v}}
\newcommand{\ww}{\mathbf{w}}

\newcommand{\pspace}{\textsc{PSpace}\xspace}
\newcommand{\nl}{\textsc{NL}\xspace}
\newcommand{\np}{\textsc{NP}\xspace}
\newcommand{\expspace}{\textsc{ExpSpace}\xspace}

\newcommand{\hp}{\mathcal{HP}}
\newcommand{\subsum}{\textsc{Subset Sum}\xspace}

\newcommand{\inc}[1]{\add{#1}{1}}
\newcommand{\dec}[1]{\sub{#1}{1}}
\newcommand{\incordec}[1]{$#1 \,\, *\!\!= \, 1$}
\newcommand{\initialise}{\textbf{initialise to} $0$}
\newcommand{\coreadd}[2]{#1 \,\, +\!\!= \, #2}
\newcommand{\coresub}[2]{#1 \,\, -\!\!= \, #2}
\newcommand{\add}[2]{$\coreadd{#1}{#2}$}
\newcommand{\sub}[2]{$\coresub{#1}{#2}$}

\newcommand{\para}[1]{\subparagraph*{#1.}}
\makeatletter
\def\BState{\State\hskip-\ALG@thistlm}
\makeatother

\algrenewcommand{\algorithmiccomment}[1]{\qquad$\rightarrow$ #1}
\newcommand{\goto}[2]{\textbf{goto} {\footnotesize #1} \textbf{or} {\footnotesize #2}}
\newcommand{\gotod}[1]{\textbf{goto} {\footnotesize #1}}

\newcommand{\halt}{\textbf{halt}}
\newcommand{\haltz}[1]{{\halt} \textbf{if} $#1 = 0$}

\newcommand{\vr}[1]{\mathsf{#1}}

\newcommand*{\eg}{e.g.\@\xspace}

\begin{document}

\maketitle

\begin{abstract}
The reachability problem is a central decision problem in verification of vector addition systems with states (VASS).
In spite of recent progress, the complexity of the reachability problem remains unsettled, and it is closely related to the lengths of shortest VASS runs that witness reachability.

We obtain three main results for VASS of fixed dimension.  For the first two, we assume that the integers in the input are given in unary, and that the control graph of the given VASS is flat (i.e., without nested cycles).  We obtain a family of VASS in dimension~$3$ whose shortest runs are exponential, and we show that the reachability problem is \np-hard in dimension~$7$.  These results resolve negatively questions that had been posed by the works of Blondin et al.\ in LICS 2015 and Englert et al.\ in LICS 2016, and contribute a first construction that distinguishes $3$-dimensional flat VASS from $2$-dimensional ones.
Our third result, by means of a novel family of products of integer fractions, shows that $4$-dimensional VASS can have doubly exponentially long shortest runs.  The smallest dimension for which this was previously known is~$14$.

\end{abstract}

\section{Introduction}
\label{sec:introduction}

\para{Context}

Vector addition systems with states (shortly, VASS) \cite[cf.\ Section~5.1]{Greibach78a}, \cite{HopcroftP79}, vector addition systems without states (shortly, VAS)~\cite{KarpM69}, and Petri nets~\cite{Petri62}, are equally expressive with well-known straightforward mutual translations.  They form a long established model of concurrency with extensive applications in modelling and analysis of hardware \cite{BurnsKY00,LerouxAG15}, software \cite{GermanS92,BouajjaniE13,KKW14} and database \cite{BojanczykDMSS11,BojanczykMSS09} systems, as well as chemical~\cite{AngeliLS11}, biological \cite{PelegRA05,BaldanCMS10} and business \cite{Aalst15,LiDV17} processes (where the references are illustrative).

Two central decision problems in the context of formal verification based on that model are the following.  Stated in terms of the first formalism, the input of both problems is a VASS~$\mathcal{V}$, and two configurations $p(\vv)$ and~$q(\ww)$.
\begin{description}
\item[Coverability]
asks whether $\mathcal{V}$ has a run starting at $p(\vv)$ and finishing at some configuration $q(\ww')$ such that $\ww' \geq \ww$.  Thus the final configuration of the run needs to have control that is in the given target state~$q$ and resources that are component-wise no smaller than the given target vector~$\ww$.  In applications, $q(\ww)$ is typically seen as a minimal unsafe configuration, and the coverability problem is fundamental for verifying safety properties.
\item[Reachability]
asks whether $\mathcal{V}$ has a run starting at $p(\vv)$ and finishing at~$q(\ww)$.  Thus the run needs to reach the given target configuration exactly.  It has turned out that verification of liveness properties amounts to solving the reachability problem~\cite{Hack74}.  Moreover, a plethora of problems from formal languages~\cite{Crespi-ReghizziM77}, logic~\cite{Kanovich95,DemriFP16,DeckerHLT14,ColcombetM14}, concurrent systems~\cite{GantyM12,EsparzaGLM17}, process calculi~\cite{Meyer09}, linear algebra~\cite{HL18} and other areas (the references are again illustrative, cf.\ Schmitz's recent survey~\cite{Schmitz16siglog}) are inter-reducible with the reachability problem.
\end{description}
\noindent
The coverability problem was found \expspace-complete already in the 1970s~\cite{lipton76,Rackoff78}, and the reachability problem was proved decidable in the early 1980s~\cite{Mayr84}.  However, the complexity of the latter has become one of the most studied open questions in the theory of verification.  The best upper and lower bounds are both very recent, and are given by an Ackermannian function~\cite{Schmitz18cigar} and a tower of exponentials~\cite{CzerwinskiLLLM19}, respectively.


\para{Fixed Dimension VASS}

The gaps in the state of the art on the complexity of the reachability problem are particularly vivid when the dimension is fixed.  For concreteness, we focus on VASS, bearing in mind that corresponding statements in terms of VAS or Petri nets can be obtained by means of standard translations (we refer to~\cite[Section~2.1]{Schmitz16siglog} for details, noting that in some cases the dimension is affected by a small additive constant).
The only broadly settled cases are for dimensions $1$ and $2$, as shown in the following table, where `unary' and `binary' specify how the integers in the input to the reachability problem are encoded.
\begin{center}
\begin{tabular}{r|c|c|}
              & unary VASS                      & binary VASS                           \\ \hline
dimension~$1$ & \nl-complete~\cite{ValiantP75}  & \np-complete~\cite{HaaseKOW09}        \\ \hline
dimension~$2$ & \nl-complete~\cite{EnglertLT16} & \pspace-complete~\cite{BlondinFGHM15} \\ \hline
\end{tabular}
\end{center}
%
For dimensions $d \geq 3$, the best known bounds are from \cite{Schmitz18cigar} and \cite{CzerwinskiLLLM19}, namely membership of the fast-growing primitive recursive class $\mathbf{F}_{d + 4}$ and hardness for $(d - 13)$-\expspace when $d \geq 13$, respectively, which hold with both unary and binary encodings.  In particular, for $3 \leq d < 13$, no better lower bounds have been known than \nl for unary VASS and \pspace for binary VASS, whereas the $\mathbf{F}_{d + 4}$ upper bound is far above elementary already for $d = 3$.


\para{Flat Control}

The structural restriction of flatness, which is essentially that the control graph contains no nested cycles, has long played a prominent role in a number of settings in verification, cf.\ e.g.~\cite{ComonC00}.  In fact, all the tight upper bounds for dimensions $1$ and $2$ recalled above can be seen as due to the effective flattability of $2$-dimensional VASS~\cite{LerouxS04}.
Regarding the complexity of reachability for flat VASS, there has been a marked contrast in the state of the art depending on the encoding.
\begin{description}
\item[Binary:]
Thanks to reducibility to existential Presburger arithmetic~\cite{FribourgO97,BlondinFGHM15}, we have \np membership, even when the dimension is not fixed.  And already for dimension~$1$, we have \np hardness.
\item[Unary:]
With the exception of dimensions $1$ and $2$ for which we have the \nl memberships, no better upper bound than \np has been known in dimension $3$ or higher.  And for any fixed dimension, no better lower bound than \nl has been obtained.
\end{description}
Interestingly, from the results of Rosier and Yen~\cite{RosierY86}, we have that the coverability problem for fixed dimension flat VASS is in \np with the binary encoding and in \nl with the unary encoding, which not provably better than the reachability problem as just discussed.


\para{Main Results}

The \nl memberships of reachability for unary VASS in dimension~$2$ and of coverability for unary VASS in any fixed dimension were obtained by proving that polynomially bounded witnessing runs always exist.  It is therefore pertinent to ask:
\moveup
\begin{quote}
\emph{Do polynomially bounded witnessing runs exist for reachability for unary flat VASS in fixed dimensions greater than~$2$?}
\end{quote}
\moveup
Our first main result, presented in Section~\ref{sec:exponential}, provides a negative answer immediately in dimension~$3$.  We believe this is very significant for the continuing quest to understand the reachability problem, for which as we have seen there is currently a huge complexity gap already in dimension~$3$.  Namely, $3$-dimensional VASS have so far been distinguished from $2$-dimensional VASS only by means of the infamous example of Hopcroft and Pansiot~\cite[proof of Lemma~2.8]{HopcroftP79}, which shows that, in contrast to the latter, the former do not have semi-linear reachability sets and are hence not flattable.  However, we now have a new distinguishing feature which is present even under the restriction of flatness.

Even if polynomially bounded witnessing runs do not exist, it is conceivable that the decision problem nevertheless has low complexity, so we next ask:
\moveup
\begin{quote}
\emph{Is reachability for unary flat VASS in \nl in fixed dimensions greater than~$2$?}
\end{quote}
\moveup
We show that this is unlikely in Section~\ref{sec:np}, where our second main result establishes \np hardness in dimension~$7$.  This provides the first concrete indication that the reachability problem is harder than the coverability problem for fixed dimension flat VASS.

Lastly, we turn to binary VASS in fixed dimensions~$d$, where without the flat assumption, the enormous complexity gap between \pspace hardness and $\mathbf{F}_{d + 4}$ membership remains for $3 \leq d \leq 13$.  Given that exponentially bounded witnessing runs exist for $d = 2$~\cite{EnglertLT16} (which yields \pspace membership) but not for $d = 14$~\cite{CzerwinskiLLLM19}, we ask:
\moveup
\begin{quote}
\emph{Do exponentially bounded witnessing runs exist for reachability for binary VASS in fixed dimensions from $3$ to~$13$?}
\end{quote}
\moveup
A negative answer is provided in Section~\ref{sec:fractions} by our third main result, which exhibits a family of $4$-dimensional VASSes whose shortest witnessing runs are doubly exponentially long.


\para{Technical Contributions}

In all three of the main results, we make use of a key technical pattern first seen in~\cite{CzerwinskiLLLM19}, namely checking divisibility of a counter~$\vr{x}$ by a large integer as follows: ensure that a counter~$\vr{y}$ is initially equal to $\vr{x}$, then multiply $\vr{x}$ weakly (which a priori may nondeterministically produce an erroneous smaller result) by many integer fractions greater $1$ whose product is $c / d$, and finally verify that $\vr{x} = \vr{y} \cdot (c / d)$ by subtracting $c$ from $\vr{x}$ and $d$ from $\vr{y}$ repeatedly until they are both~$0$.  The divisibility by the large integer is ensured because the check succeeds if and only if the weak multiplications are all exact.  However, much additional development has been involved:
\begin{enumerate}
\item
For the exponentially long shortest runs in Section~\ref{sec:exponential}, we employ the factorial fractions also seen in~\cite{CzerwinskiLLLM19}, but in reverse order, with the construction stripped to its essentials to minimise the dimension, and with a detailed divisibility analysis of large integers.
\item
The \np hardness in Section~\ref{sec:np} builds on the development in the previous section, adding careful machinery that facilitates exact computations on exponentially large integers.
\item
To obtain the doubly exponentially long shortest runs in Section~\ref{sec:fractions}, we have developed an intricate new family of sequences of fractions, where in contrast to the much simpler factorial equations, the number of distinct fractions in a sequence is logarithmic in relation to both the numerators and the denominators as well as to the length of the sequence.
\end{enumerate}




\section{Preliminaries}
\label{sec:preliminaries}

\para{Vector Addition Systems with States}
A \emph{vector addition system with states} in dimension $d$ ($d$-VASS, or simply VASS if the dimension is irrelevant) 
is a pair $\mathcal{V} = (Q,T)$ consisting of a finite set $Q$ of states and a finite set of transitions $T \subseteq Q \times \Z^d \times Q$. 
The size of a VASS is $|Q| + |T| \cdot s$, where $s$ is the maximum on the representation size of a vector in $T$.
A \emph{configuration} of a $d$-VASS is a pair $(p,\vv) \in Q \times \N^d$, denoted $p(\vv)$, consisting of a state $p$ and
a nonnegative integer vector $\vv$. 
A run of a $d$-VASS is a sequence of configurations
\moveup\moveup
\begin{align} \label{eq:run}
p_0(\vv_0), \ \ldots, \ p_k(\vv_k),
\end{align}
such that for every $1\leq i \leq k$ there is a transition $\alpha_i = (p_{i-1},\ww_i,p_{i})$ $\in T$ satisfying $\vv_{i-1} + \ww_i = \vv_{i}$.
The sequence of transitions
$\alpha_1, \ \ldots, \ \alpha_k$
we call the \emph{path} of the run~\eqref{eq:run}.

We are interested in the complexity of the \emph{reachability problem}: given a $d$-VASS and two configurations $p(\vv)$, $q(\ww)$ does there exist a run from $p(\vv)$ to $q(\ww)$. W.l.o.g.~we can restrict $\vv = \ww = 0$ to be the zero vectors, 
as the general case polynomially reduces to such restricted case. Indeed, it suffices to add a new initial state whose only out-going
transition adds $\vv$, and likewise a new final state whose only in-going transition subtracts $\ww$.
In the sequel we usually assume that VASS is additionally equipped with a pair of configurations,
a source $p(\vv)$ and a target $q(\ww)$, thus $\mathcal{V} = (Q,T, p(\vv), q(\ww))$. 
Thus we do not distinguish between a VASS and a VASS reachability instance.
Runs from
$p(\vv)$ to $q(\ww)$ we call  \emph{halting runs} of $\mathcal V$.

We study the reachability problem under two further restrictions. 
The first restriction assumes that the dimension $d$ is fixed. In this case it may matter,
for the complexity of the reachability problem, whether the numbers appearing in the vectors in $T$ are encoded in unary or binary. 
We will thus distinguish these two cases, and speak of unary, respectively binary VASS.
Note that in the unary case one can assume w.l.o.g.~all vectors in $T$ to be either the zero vector, or 
the unit vector $e_i = (0,\ldots,0,1,0,\ldots,0)$ with single $1$ on some $i$-th coordinate, or inverse $-e_i$ thereof. 
The second restriction is \emph{flatness} and concerns cycles in runs (see \eg~\cite{LerouxS04,BlondinFGHM15}). 
The path $\alpha_1, \ \ldots, \ \alpha_k$ of a
run~\eqref{eq:run} is called \emph{simple} if there is no repetition of states along the path;
it is called \emph{simple cycle} if there is no repetition of states along the path except for the first and the last states which are equal: 
$p_0 = p_k$. 
A VASS is \emph{flat} if every state admits at most one simple cycle
on it (i.e., the VASS has no nested cycles).

%
%


\para{Counter Programs}
We are going to represent VASSes by counter programs.
A \emph{counter program} is a numbered sequence of commands of the following types:
\moveup
\begin{quote}
\begin{tabular}{l@{\qquad}l}
\add{\vr{x}}{n}                       & (increment counter~$\vr{x}$ by $n$) \\
\sub{\vr{x}}{n}                       & (decrement counter~$\vr{x}$ by $n$) \\
\goto{$L$}{$L'$}                   & (jump to either line~$L$ or line~$L'$)
\end{tabular}
\end{quote}
\moveup
except that the first and the last command of the program, respectively, are of the form
\moveup
\begin{quote}
\begin{tabular}{l@{\qquad}l}
\initialise & (initialise all counters to zero); \\
\haltz{\vr{x_1}, \ldots, \vr{x}_l} & (terminate provided all listed counters are zero).
\end{tabular}
\end{quote}
\moveup
(We note that in the unary case, increments \add{\vr{x}}{m} and decrements \sub{\vr{x}}{m} can be written as $m$ 
consecutive unitary increments \inc{\vr{x}} and decrements \dec{\vr{x}}, respectively, introducing only linear blow-up.
In the binary case this would lead to an exponential blow-up.)
Indeed, a counter program $\mathcal{P}$ represents a VASS (in fact, a VASS reachability instance) 
of dimension equal to the number of counters used in $\mathcal{P}$,
with a separate state for every line in $\mathcal{P}$.
The increment and decrement commands in $P$ are simulated by transition vectors of the VASS. 
The source and target configurations of the VASS correspond to the first and last line of
$\mathcal{P}$.
The size of the VASS is linear with respect to the size of the program.
This convenient representation was adopted e.g.~in~\cite{Esparza98,CzerwinskiLLLM19}.

Accordingly with runs of a VASS, we speak of runs of a counter program (in particular, values of counters along a run are nonnegative)
with the proviso that the initial value of all counters is 0.
A run is \emph{halting} if it has successfully executed its (necessarily last) {\halt} command; otherwise, the run is \emph{partial}.
The reachability problem for a VASS translates into the question whether there exists a halting run in a counter program.

Note that a counter program does not need to test for zero all counters in the final {\halt} command;
for the sake of presentation it is convenient to allow for halting runs with non-zero final value of certain (irrelevant) counters.
On the other hand,
formally, our intention is that a counter program represents a VASS reachability instance 
with the zero target vector.
This incompatibility can be circumvented by assuming that counter programs are implicitly
completed with additional loops allowing 
to decrease every untested counter just before executing the \textbf{halt} command.



In case of fragments of counter programs which neither start with \textbf{initialise} nor end with \textbf{halt},
we consider explicit \emph{initial} and \emph{final} values of counters.
Note however that due to nondeterministic \textbf{goto} command, final values are not uniquely determined by initial ones.

When writing counter program we use a syntactic sugar:
we write \gotod{$L$} instead of \goto{$L$}{$L$}, and 
whenever a program repeats the block of commands in line \ref{l:testz.iter} 
some nondeterministically chosen number of times (possibly zero, possibly infinite), 
as shown on the left, we use a shorthand as shown on the right:

\twocol{.40}{
\movedown\movedown
\begin{quote}
\begin{algorithmic}[1]
\State \goto{\ref{l:testz.exit}}{\ref{l:testz.iter}} \label{l:testz.rep}
\State <iterated commands>\label{l:testz.iter}
\State \gotod{\ref{l:testz.rep}}
\State <remaining commands>                      \label{l:testz.exit}
\end{algorithmic}
\end{quote}
}
{\qquad\qquad}{.50}{
\begin{quote}
\begin{algorithmic}[1]
\Loop
  \State <iterated commands>
\EndLoop
\State <remaining commands>
\end{algorithmic}
\end{quote}
}
\movedown\movedown

\noindent
In the sequel we will only occasionally use \textbf{goto} commands \emph{explicitly}.
Observe that a counter program without explicit \textbf{goto} commands, but using \emph{unnested} 
\textbf{loop} commands (which \emph{implicitly} use \textbf{goto} commands), always represents a flat VASS.

\twocol{.35}{
\moveup
\begin{algorithm}[H]
\caption{Weak multiplication by $\frac{c}{d}$, for $c > d$.}
\label{alg:weak-multiplication}
\begin{algorithmic}[1]
\Loop
\State \sub{\vr{x}}{1} \quad \add{\vr{y}}{1}
\EndLoop
\Loop \label{l:weak-second-loop}
\State \add{\vr{x}}{c} \quad \sub{\vr{y}}{d}
\EndLoop
\end{algorithmic}
\end{algorithm}
\movedown
}
{\qquad}{.55}{
We end this section with examples of counter programs that
\emph{weakly} compute a number $b$ in some counter $\vr{x}$, i.e.,
all runs end with $\vr{x} \leq b$, and there is a run that ends with $\vr{x} = b$.
On the way we also introduce macros to be used later to facilitate writing complex programs.
As a preparation, consider the program in Algorithm~\ref{alg:weak-multiplication} which weakly multiplies  
the initial value of $\vr{x}$ by $\frac{c}{d}$.
}

\moveup\moveup\noindent
Let $x_0, y_0$ and $x_1, y_1$ be initial and final values, respectively, of counters $\vr{x}$, $\vr{y}$.
We claim that the sum of final values is at most $\frac{c}{d}$ times larger than the sum of initial values.
Moreover, it is exactly $\frac{c}{d}$ times larger if, and only if, both loops are iterated \emph{maximally}: 
the first loop exits only when the counter $\vr{x}$, decreased in its every iteration, reaches the minimal possible value $0$;
and likewise the second loop exits only when the counter $\vr{y}$ reaches $0$.
Enforcing maximal iteration of loops will be our fundamental technical objective in the sequel.
\begin{claim}\label{claim:weak-multiplication}
Let $x', y'$ be the values of counters $\vr{x},\vr{y}$ at the exit from the first loop.
Then $x_1 + y_1 \leq (x_0 + y_0) \cdot \frac{c}{d}$. Moreover, $x_1 = (x_0 + y_0) \cdot \frac{c}{d}$ if, and only if,
$x' = y_1 = 0$.
\end{claim}


\moveup\moveup\moveup
\twocol{.60}{
\begin{algorithm}[H]
\caption{Program fragment $\weakexp{b}{\vr{x}}{\vr{x}'}$.} 
\label{alg:weakly}
\begin{algorithmic}[1]
\State \initialise
\For {\, $i$ \, := \, $m$ \, \textbf{downto } $0$}\label{l:for}
\Loop \label{l:weakly1}
\State \sub{\vr{x}}{1} \quad \add{\vr{y}}{1}
\EndLoop
\Loop
\State \add{\vr{x}}{2} \quad \sub{\vr{y}}{1}
\EndLoop \label{l:weakly2}
\If{$b_i = 1$}{\, \inc{\vr{x}}}\label{l:if}
\EndIf
\EndFor
\end{algorithmic}
\end{algorithm}
\moveup\moveup
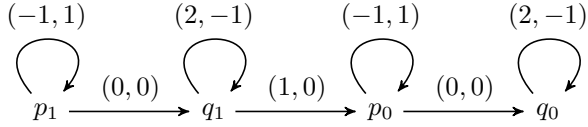
\begin{figure}[H]

\begin{tikzpicture}[->,>=stealth',shorten >=1pt,auto,node distance=2.2cm,semithick]

\node (p0) {$p_1$};
\node (q0) [right of=p0] {$q_1$};
\node (p1) [right of=q0] {$p_0$};
\node (q1) [right of=p1] {$q_0$};

%



\path[->]
(p0) edge[->, in=50, out=130, min distance=0.5cm,loop] node {$(-1,1)$}(p0)
(q0) edge[->, in=50, out=130, min distance=0.5cm,loop] node {$(2,-1)$}(q0)
(p1) edge[->, in=50, out=130, min distance=0.5cm,loop] node {$(-1,1)$}(p1)
(q1) edge[->, in=50, out=130, min distance=0.5cm,loop] node {$(2,-1)$}(q1)
(p0) edge[->] node[above] {$(0,0)$} (q0)
(q0) edge[->] node[above] {$(1,0)$} (p1)
(p1) edge[->] node[above] {$(0,0)$} (q1)
%
;
\end{tikzpicture}
\caption{A 2-VASS represented by the program $\weakexp{2}{\vr{x}}{\vr{x}'}$. The first coordinate corresponds to the value of counter $\vr{x}$ and the second one to the value of counter $\vr{y}$.}\label{fig:vassweak}
\end{figure}
\moveup\moveup
}
{\qquad}{.30}{ \vspace{-24mm}
\begin{algorithm}[H]
\caption{Unfolding of macros in $\weakexp{2}{\vr{x}}{\vr{x}'}$.} 
\label{alg:unf}
\begin{algorithmic}[1]
\State \initialise

\Loop 
\State \sub{\vr{x}}{1} \quad \add{\vr{y}}{1}
\EndLoop
\Loop
\State \add{\vr{x}}{2} \quad \sub{\vr{y}}{1}
\EndLoop 

\State \inc{\vr{x}}

\Loop 
\State \sub{\vr{x}}{1} \quad \add{\vr{y}}{1}
\EndLoop
\Loop
\State \add{\vr{x}}{2} \quad \sub{\vr{y}}{1}
\EndLoop 
\end{algorithmic}
\end{algorithm}
}
The counter program $\weakexp{b}{\vr{x}}{\vr{x}'}$ shown in~Algorithm~\ref{alg:weakly} weakly computes a number $b$,
assuming that $b_m \dots b_0 = \bin{b}$ is the binary representation of $b$ (the oldest bit $b_m = 1$).
The \textbf{halt} command is omitted as no zero-testing is relevant in this example.
We use \textbf{for} and \textbf{if then} preprocessing macros with the following semantics.
The macro
\moveup
\begin{quote}
\begin{algorithmic}
\For {\, $i$ \, := \, $m$ \, \textbf{downto } $0$}
\quad <program fragment>
\EndFor
\end{algorithmic}
\end{quote}
\moveup
is understood as $(m+1)$-fold repetition of copies of <program fragment>:
\moveup
\begin{quote}
\begin{algorithmic}
  \State <program fragment>  $\qquad (i = m)$
  \State <program fragment>  $\qquad (i = m-1)$
  \State $\ldots$
  \State <program fragment>  $\qquad (i = 0)$
\end{algorithmic}
\end{quote}
for $i = m, m-1, \ldots, 0$.
It is important that $i$ is not a counter but a meta-variable that is treated as a constant in every program fragment. By convention we use different fonts for counters and meta-variables:
$\vr{i}$ is a counter while $i$ is a meta-variable.
Furthermore in every copy, say for $i=k$,
%
%
at every appearance of the macro
{\bf if} $\varphi(i)$ {\bf then} <optional program fragment>,
%
%
the formula $\varphi(i)$ is evaluated and, if it evaluates positively then macro is replaced by  
<optional program fragment>, otherwise it is removed.
%
%
Specifically, consider for example $m = 1$ and $b = 2$, i.e.,  $b_1 = 1$ and $b_0 = 0$. 
Unfolding of the macros appearing in $\weakexp{2}{\vr{x}}{\vr{x}'}$ yields the
counter program shown in Algorithm~\ref{alg:unf}. Figure~\ref{fig:vassweak} shows the corresponding 2-VASS.
We remark that the programs in Algorithm~\ref{alg:weak-multiplication} and Algorithm~\ref{alg:weakly} 
represent flat VASS.

Clearly, we do not want \textbf{for} and \textbf{if then} to be full-fledged commands operating on program counters, as this would make counter programs as powerful as Minsky machines, hence undecidable. 
They are just pre-processing macros that operate on meta-variables $i$ only, and constitute 
syntactic sugar helpful in writing repetitive program fragments.

\begin{proposition}\label{prop:weak}
The program $\weakexp{b}{\vr{x}}{\vr{x}'}$ weakly computes $b$. 
\end{proposition}
\begin{proof}
By Claim~\ref{claim:weak-multiplication}, the program $\weakexp{b}{\vr{x}}{\vr{x}'}$  weakly multiplies 
$\vr{x}$ by $2$ in lines~\eqref{l:weakly1}--\eqref{l:weakly2}. 
Combining this with addition of a bit in~\eqref{l:if} gives weak computation of $b$.
\end{proof}

\section{Exponential Shortest Runs}
\label{sec:exponential}

\begin{theorem}\label{th:exponential}
There is a family of unary flat 3-VASS $(\mathcal{V}_n)_{n\in\N}$ of size $\O(n^2)$ such that every 
halting run of $\mathcal{V}_n$ is of length exponential in $n$.
%
\end{theorem}

\moveup\moveup
\twocol{.40}{
\begin{algorithm}[H]
\caption{Counter program $\mathcal{P}_n$.}
\label{alg:exponential}
\begin{algorithmic}[1]
\State \initialise
\State \inc{\vr{x}} \quad \inc{\vr{y}}
\Loop
  \State \inc{\vr{x}} \quad \inc{\vr{y}}
\EndLoop
\For {\, $i$ \, := \, $n$ \, \textbf{down to } $1$}\label{l:forstart1}
\Loop \label{l:xz}
\State \sub{\vr{x}}{1} \quad \add{\vr{z}}{1} \label{l:xz2}
\EndLoop \label{l:xzend}
\Loop \label{l:zx}
\State \add{\vr{x}}{i+1} \quad \sub{\vr{z}}{i} \label{l:forend1}
\EndLoop \label{l:zxend}
\EndFor
\Loop
\State \sub{\vr{x}}{n+1} \quad \sub{\vr{y}}{1} \label{l:loopy}
\EndLoop
\State \haltz{\vr{y}}.\label{l:halt}
\end{algorithmic}
\end{algorithm}
\movedown
}
{\qquad}{.50}{
In this section we prove the theorem.
The VASS $\mathcal{V}_n$ are represented by the counter programs $\mathcal{P}_n$ 
shown in Algorithm~\ref{alg:exponential}. 
The idea of multiplying by consecutive fractions $\frac{2}{1}$, $\frac{3}{2} \ldots \frac{n+1}{n}$ 
comes from~\cite{CzerwinskiLLLM19} (cf.~Algorithms~I,II therein), however, we need to apply the multiplications in the reverse order. 
The size of $\mathcal{P}_n$ is quadratic in $n$, as the \textbf{for} macro unfolds $n$ times, and
the constants appearing in the increment/decrement commands, like $i+1$ in \add{\vr{x}}{i+1}, are written in unary.
Consider any run that reaches (but not yet executes) line~\ref{l:halt}.
For every $i = n, \ldots, 1$ let $x_i$ and $z_i$ be the values of counters 
$\vr{x}$ and $\vr{z}$, respectively, at the exit from the loop in lines~\ref{l:zx}--\ref{l:zxend}.
}
Similarly, let $x_i'$ be the value of counter $\vr{x}$ at the exit from the loop in lines~\ref{l:xz}--\ref{l:xzend}. 
For uniformity we write $x_{n+1}$ and $z_{n+1}$ for the values of $\vr{x}$ and $\vr{z}$, respectively, 
just before entering the \textbf{for} macro, and call these values initial. 
Notice that $x_{n+1}$ is equal to the value of counter $\vr{y}$ at that point and $z_{n+1} = 0$.
By Claim~\ref{claim:weak-multiplication} we derive:

\begin{claim}\label{claim:fraction}
For all $i = 1,\ldots, n$, we have
$
x_i + z_i \le (x_{i+1} + z_{i+1}) \cdot \frac{i+1}{i}.
$
\end{claim}
%
%
We focus on runs that \emph{maximally iterate} both inner loops, by which we mean:
\begin{itemize}
\item the value of $\vr{x}$ is 0 at the exit of the loop in lines \ref{l:xz}--\ref{l:xzend};
\item the value of $\vr{z}$ is 0 at the exit of the loop in lines \ref{l:zx}--\ref{l:zxend};
\end{itemize}

\begin{claim}\label{claim:mult}
We have $x_1 \le x_{n+1} \cdot (n+1)$. The equality holds if, and only if, $z_i = x'_i = 0$ for all $i = 1,\ldots,n$.
\end{claim}
For a finite subset $X\subseteq \N$ of natural numbers, we write $\lcmpar X$ for the least common multiple of all numbers in $X$.
We will use the number $N(n)$ defined as
\[
N(n) \quad := \quad \frac{\lcmpar {\set{2, \dots, n+1}}}{n+1}.
\]
%
The following two claims conclude the proof of Theorem~\ref{th:exponential}. 

\begin{claim}\label{claim:lcm}
The function $N$ grows exponentially with respect to $n$.
The binary representation $\bin{N(n)}$ is computable in time polynomial with respect to $n$.
\end{claim}
\begin{proof}
For the exponential upper bound we recall that $N(n) \le n!$
For the exponential lower bound we use the prime number theorem (proved independently by Jacques Hadamard and Charles Jean de la Vall\'ee Poussin in 1896): the numer of primes $\pi(n)$ between $2$ and $n$ is at least $\pi(n) \ge c \cdot n^{\epsilon} - 1$ for some constants $c > 0$ and $0 < \epsilon < 1$. 
Since $\lcmpar{\set{2 \dots n+1}}$ must be divisible by all prime
numbers between $2$ and $n+1$ and each prime number is at least $2$ we
get $\lcmpar{\set{2, \ldots, n+1}} \ge 2^{\pi(n+1)}$.

$N(n)$ is computated by exhaustive enumeration of all non-divisors of $n+1$, computing their prime decompositions,
and combining them into prime decomposition of $N(n)$.
\end{proof}
\begin{claim}\label{claim:final}
For every initial value $\vr{x}_{n+1}$ of counter $\vr{x}$ there is at most one halting run. 
Such a run exists if, and only if, $\vr{x}_{n+1}$ is a positive multiple of $N(n)$.
\end{claim}

\begin{proof}
Recall that the last loop in $\mathcal{P}_n$ in line~\ref{l:loopy} decreases simultaneously $\vr{y}$ by $1$ and $\vr{x}$ by $n+1$. 
Therefore, the run halts only if $\vr{x} \ge \vr{y} \cdot (n+1)$. By Claim~\ref{claim:mult} we have $\vr{x} \le \vr{y} \cdot (n+1)$.
Thus every halting run satisfies the equality $\vr{x} = \vr{y} \cdot (n+1)$.
By Claim~\ref{claim:mult} we know that is possible only if $\vr{z}_i = \vr{x}_i' = 0$ for all $i = 1,\ldots,n$, which uniquely determines the run for a given $\vr{x}_{n+1}$.
It remains to prove that a halting run exists if, and only if, $\vr{x}_{n+1}$ is a positive multiple of $N(n)$. 
Notice that by Claim~\ref{claim:fraction} and Claim~\ref{claim:mult} in a halting run
$$
\vr{x}_i = \vr{x}_{i+1} \cdot \frac{i+1}{i} = \ldots = \vr{x}_{n+1} \cdot \prod_{j = i}^{n} \frac{j+1}{j} = \vr{x}_{n+1} \cdot \frac{n+1}{i}.
$$
Therefore $\vr{x}_{n+1} \cdot (n+1)$ must be always divisible by all
numbers in ${\set{1 \dots n}}$. Since $n+1$ divides $\vr{x}_{n+1}
\cdot (n+1)$ as well, we deduce that $\lcmpar{\set{2, \ldots, n+1}}$
divides $\vr{x}_{n+1}\cdot (n+1)$ which is equivalent to $N(n)$
divides $\vr{x}_{n+1}$. 
Conversely, if $\vr{x}_{n+1}$ is a multiple of $N(n)$ then there is a run where all loops are iterated
maximally, satisfying $\vr{x}_i = \vr{x}_{i+1} \cdot \frac{i+1}{i}$ and thus halting.
\end{proof}

\section{NP-hardness}
\label{sec:np}

This section is devoted to proving \NP-lower bound for flat VASS in fixed dimension.

\begin{theorem}\label{th:np}
The reachability problem for unary flat 7-VASS is \NP-hard.
\end{theorem}
As mentioned in the introduction NP-membership is already known (even in binary VASS of unrestricted dimension). 
Thus as a corollary we get the following result.

\begin{corollary}
The reachability problem for flat $d$-VASS is \NP-complete for fixed $d \ge 7$.
\end{corollary}

To prove Theorem~\ref{th:np}
we reduce from the \subsum problem:  given a set of positive integers $S = \set{s_1,\ldots, s_k}\subseteq \N-\set{0}$ 
and an integer $s_0 > 0$, determine 
if some subset $R \subseteq S$ satisfies $\sum_{s \in R}s = s_0$. 
Note that all the numbers $s_0,s_1, \ldots, s_k$ are encoded in binary.

Fix an instance $s_0,s_1, \ldots, s_k$ of the \subsum problem and let
$n$ be the smallest natural number such that $N(n) \ge s_0,s_1,\ldots, s_k$. 
By Claim~\ref{claim:lcm} the number $n$ as well as the binary representation 
$\bin {N(n)} = b_m \dots b_0$ is computable in time polynomial 
with respect to the sizes of binary representations of $s_0,s_1, \ldots, s_k$.
Recall that $b_m = 1$. 
We are going to define a unary counter program $\mathcal P$ of polynomial size using $7$ counters, 
as a function of  $s_0,s_1, \ldots, s_k$ and $n$,
which halts if, and only if, the instance $s_0,s_1, \ldots, s_k$ is positive.

The main obstacle is that the numbers in the \subsum problem are represented in binary, while
the numbers in a counter program 
are to be represented in unary.
Thus we have to exactly compute with exponential numbers,
using a fixed number of $7$ counters.

\para{Construction of $\mathcal P$}
We face the challenge by combining the weak computation given by  Proposition~\ref{prop:weak} (that allows us to compute 
\emph{at most} a required value $b$) with the insight of the proof of Theorem~\ref{th:exponential} 
(that enforces that the computed value is simultaneously \emph{at least} $b$). 

\twocol{.55}{
\moveup\moveup\moveup\moveup\moveup\moveup
\begin{algorithm}[H]
\caption{Counter program $\mathcal{I}$. 
}
\label{alg:initialise}
\begin{algorithmic}[1]
\State \initialise \label{l:startI}
\State \inc{\vr{x}} \quad \inc{\vr{y}} \quad \inc{\vr{e}} \quad \add{\vr{f}}{k+1}
\For {\, $i$ \, := \, $m-1$ \, \textbf{to } $0$ \,}
\Loop
\State \sub{\vr{x}}{1} \quad \add{\vr{x}'}{1}
\State \sub{\vr{y}}{1} \quad \dec{\vr{e}} \quad \sub{\vr{f}}{k+1}
\EndLoop
\Loop
\State \add{\vr{x}}{2} \quad \sub{\vr{x}'}{1}
\State \add{\vr{y}}{2} \quad \add{\vr{e}}{2} \quad \add{\vr{f}}{2(k+1)}
\EndLoop 
\If{$b_i = 1$}{\\ \qquad\quad \inc{\vr{x}} \ \inc{\vr{y}} \ \inc{\vr{e}} \ \add{\vr{f}}{k+1}}
\EndIf
\EndFor \label{l:l8}
\For {\, $i$ \, := \, $n$ \, \textbf{down to } $1$ \,} \label{l:forI}
\Loop
\State \sub{\vr{x}}{1} \quad \add{\vr{z}}{1}
\EndLoop
\Loop
\State \add{\vr{x}}{i+1} \quad \sub{\vr{z}}{i}
\EndLoop
\EndFor
\Loop
\State \sub{\vr{x}}{n+1} \quad \sub{\vr{y}}{1}
\EndLoop
\State \haltz{\vr{y}} \qquad\qquad\qquad  // \emph{ removed in $\mathcal{I}'$} 
\label{l:haltI}
\end{algorithmic}
\end{algorithm}
}
{\ }{.38}{
Program $\mathcal{I}$ in Algorithm~\ref{alg:initialise} implements this idea. 
The first half of the program, namely lines \ref{l:startI}--\ref{l:l8}, weakly computes
in counter $\vr{e}$ the value $N(n)$, and in counter $\vr{f}$ the value $N(n) \cdot (k+1)$, very much like
the counter program $\weakexp{b}{\vr{x}}{\vr{x}'}$.
Note a slight difference compared to Algorithm~\ref{alg:weakly}:
the oldest bit $b_m = 1$ is treated in a different way than other bits $b_i$ for $0\leq i< m$, by initializing counters $\vr{e}$ and $\vr{f}$ to $1$ and $k+1$, respectively,
which excludes a trivial halting run that would never iterate any loop and end with the value of $\vr{y}$ equal $0$. 
Then the second part of $\mathcal{I}$ checks, very much like the counter program $\mathcal{P}_n$,
if the values are computed exactly.
(Notice that lines \eqref{l:forI}--\eqref{l:haltI} are exactly the same as lines~\ref{l:forstart1}--\ref{l:halt} of Algorithm~\ref{alg:exponential}.)
Using Claim~\ref{claim:final} we get:
}
\begin{claim}\label{claim:intro}
Counter program $\mathcal{I}$ has exactly one halting run that computes $N(n)$ and $N(n) \cdot (k+1)$ in counters $\vr{e}$ and $\vr{f}$,
respectively, and $0$ in the remaining counters $\vr{x}, \vr{x}', \vr{y}$ and $\vr{z}$.
\end{claim}
The program $\mathcal P$ (shown in Algorithm~\ref{alg:P}) 
consists of the program $\mathcal{I}'$ obtained from $\mathcal{I}$ by
removing the last \textbf{halt} command.
The remaining part of $\mathcal P$
exploits the values of counters $e$ and $f$ computed by $\mathcal{I}'$ to turn weak computations of exponential numbers
into exact ones.
It never modifies the counter $\vr{y}$ again, hence $\vr{y}$ is listed in the final \textbf{halt} command of $\mathcal P$,
and uses a distinguished counter $\vr{u}$, initially set to $0$, 
a program fragment $\mathcal{R}^+_{{s_0}, \true}$, and a number of program
fragments $\mathcal{R}^-_{s, p}$ for $s\in \set{s_1, s_2 \dots s_k}$ and $p\in \set{\true,\false}$.
We call these program fragments \emph{components}.
In every halting run of $\mathcal P$, the component $\mathcal{R}^-_{s, \true}$ decrements $\vr{u}$ by $s$ while the other
component $\mathcal{R}^-_{s, \false}$ has no effect on counter $\vr{u}$.
Likewise, the component $\mathcal{R}^+_{s_0, \true}$ increments $\vr{u}$ by $s_0$.
Finally, $\vr{u}$ is zero-tested by the final \textbf{halt} command.

\twocol{.35}{
\begin{algorithm}[H]
\caption{Program $\mathcal{P}$.}
\label{alg:P}
\begin{algorithmic}
\State $\mathcal{I}'$
\State $\mathcal{R}^+_{s_0, \true}$
\State \goto{$f_1$}{$t_1$}
\State {\footnotesize $f_1$}:\ \ $\mathcal{R}^-_{s_1, \false}$ \quad {\goto{$f_2$}{$t_2$}}
\State {\footnotesize $t_1$}:\ \ $\mathcal{R}^-_{s_1, \true}$ \quad {\,\goto{$f_2$}{$t_2$}}
\State {\footnotesize $f_2$}:\ \ $\mathcal{R}^-_{s_2, \false}$ \quad {\goto{$f_3$}{$t_3$}}
\State {\footnotesize $t_2$}:\ \ $\mathcal{R}^-_{s_2, \true}$ \quad {\,\goto{$f_3$}{$t_3$}}
\State \qquad\  \ldots
\State {\footnotesize $f_k$}:\ \ $\mathcal{R}^-_{s_k, \false}$ \ \ \ {\gotod{$h$}}
\State {\footnotesize $t_k$}:\ \ $\mathcal{R}^-_{s_k, \true}$ 
\State {\footnotesize $h$}:\ \ \ \haltz{\vr{y},\vr{u},\vr{f}} \label{l:haltgoto}
\end{algorithmic}
\end{algorithm}
\vspace{0.5mm}
}
{\qquad}{.55}{
For $1\leq i \leq k$, we use {\footnotesize $f_i$}, respectively {\footnotesize $t_i$},
 to denote the the first line of the program fragment $\mathcal{R}^*_{s_i, \false}$, respectively
$\mathcal{R}^*_{s_i, \true}$. 
Every component $\mathcal{R}^*_{s_i, p}$, for $0 \leq i < k$ and $* \in \set{+,-}$, is followed by 
\goto{$f_{i+1}$}{$t_{i+1}$}.
Thus for every $i = \set{1\dots k}$ either $\mathcal{R}^-_{s_i,\false}$ or $\mathcal{R}^-_{s_i, \true}$ is executed,
as shown by the following control flow diagram of $\mathcal P$:

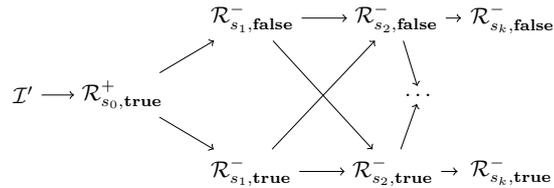
\movedown
\scalebox{0.85}{
\begin{tikzpicture}
\node (init) {$\mathcal{I}'$};

\node[right = 0.5cm of init] (R) {$\mathcal{R}^+_{s_0, \true}$};

\node[above right = 0.5cm and 0.5cm of R] (R10) {$\mathcal{R}^-_{s_1,\false}$};
\node[below right = 0.5cm and 0.5cm of R] (R11) {$\mathcal{R}^-_{s_1,\true}$};

\node[right = 0.7cm of R10] (R20) {$\mathcal{R}^-_{s_2,\false}$};
\node[right = 0.7cm of R11] (R21) {$\mathcal{R}^-_{s_2,\true}$};

\node[right = 5.5cm of init] (dots) {\ldots};

\node[above right = 0.7cm and 0.3cm of dots] (Rk0) {$\mathcal{R}^-_{s_k,\false}$};
\node[below right = 0.7cm and 0.3cm of dots] (Rk1) {$\mathcal{R}^-_{s_k,\true}$};

\path
(init) edge[->] (R)
(R) edge[->] (R10)
(R) edge[->] (R11)
(R10) edge[->] (R20)
(R11) edge[->] (R21)
(R10) edge[->] (R21)
(R11) edge[->] (R20)
(R20) edge[->] (dots)
(R21) edge[->] (dots)
(R20) edge[->] ++(1,0)
(R21) edge[->] ++(1,0)
;
\end{tikzpicture}

}

}

\noindent
Observe that every halting run of $\mathcal P$ determines a subset $R \subseteq \{1, \ldots, k\}$ such that
for $i \in R$ the component $\mathcal{R}^-_{s_i, \true}$ is executed, while for 
$i \not\in R$ the component $\mathcal{R}^-_{s_i, \false}$ is executed.

\para{The Components} 
The component $\mathcal{R}^*_{a, p}$ is shown in Algorithm~\ref{alg:component}.
By $\binext{m}{a} = a_m  \dots a_0$ 
we mean the $(m+1)$-bit binary representation of the number $a < 2^{m+1}$, 
padded with leading 0 bits if needed.
The aim of every $\mathcal{R}^*_{a, \true}$ is to increment (when $*=+$) or decrement (when $*=-$) $a$ from 
the counter $\vr{u}$, using the counters $\vr{e}$ and $\vr{f}$ to enforce exactness.
After the auxiliary counter $\vr{v}$ is initialised to $1$, in every iteration of the \textbf{for} loop 
(in lines 2-8) counter $\vr{v}$
is weakly multiplied by $2$, so after $i$ iterations its value is at most $2^i$.

\twocol{.45}{
\begin{algorithm}[H]
\caption{Component $\mathcal{R}^*_{a,p}$.  
}
\label{alg:component}
\begin{algorithmic}[1]
\State \inc{\vr{v}}
\For {\, $j$ \, := \, $0$ \, \textbf{to } $m-1$ \,}\label{l:for1}
\Loop \label{l:weak1}
\State \dec{\vr{v}} \quad \inc{\vr{v}'}
\If{$b_j = 1$}{}
\State {\, \dec{\vr{e}} \ \inc{\vr{e}'} \ \dec{\vr{f}}}\label{l:dec1}
\EndIf
\If{$p \wedge (a_j = 1)$}{\,  \incordec{\vr{u}}} \label{l:uuu}
\EndIf
\EndLoop
\Loop \label{l:bits}
\State \add{\vr{v}}{2} \quad \dec{\vr{v}'} \label{l:weak2}
\EndLoop
\EndFor
\Loop \label{l:mbitS}
\State \dec{\vr{v}} 
\State \dec{\vr{e}} \quad \inc{\vr{e}'} \quad \dec{\vr{f}}\label{l:dec2}
\If{$p \wedge (a_m = 1)$}{\,  \incordec{\vr{u}}} \label{l:mbitE}
\EndIf
\EndLoop
\Loop \label{l:loop1}
\State \inc{\vr{e}} \quad \dec{\vr{e}'} \label{l:laste}
\EndLoop
\end{algorithmic}
\end{algorithm}
\movedown
}
{\qquad}{.45}{
In lines~\ref{l:dec1} and~\ref{l:dec2} counter $\vr{f}$ is decremented, in both cases together with counter $\vr{e}$,
hence the total decrement of $\vr{f}$ is at most the initial value of counter $\vr{e}$.
Now recall that in a halting run of $\mathcal P$ the values of $\vr{e}$ and $\vr{f}$ output by $\mathcal{I}'$ are $N(n)$ and
$N(n) \cdot (k+1)$, respectively. As every halting run of $\mathcal P$ 
passes through exactly $k+1$ components and $\vr{f}$ is zero-tested by the final \textbf{halt} command of $\mathcal P$, 
every of the components forcedly decrements $\vr{f}$ by exactly $N(n)$.
Also forcedly, after $i$ iterations of the for loop in lines~\ref{l:for1}-\ref{l:weak2} the value of counter $\vr{v}$ is exactly $2^i$. 
This in consequence implies that the counter $\vr{u}$ is incremented (respectively, decremented) in lines~\ref{l:uuu} and~\ref{l:mbitE} by
exactly $a$ times, hence by $a$ in total.
Lines~\ref{l:loop1}-\ref{l:laste} are to revert the roles of counters $e$ and $e'$.
}

\noindent
Note that the oldest bit $a_m$, irrespectively of its value $0$ or $1$, is treated differently (in lines \ref{l:mbitS}-\ref{l:mbitE}) 
from the other bits $a_{m-1} \dots a_0$ of $\binext{m}{a}$ (treated in the body of the \textbf{for} loop in  lines \ref{l:bits}--\ref{l:weak2}).
This is because the auxiliary counter $\vr{v}$ needs to be multiplied by $2$ exactly $m$ times, 
which happens in the course of $m$ iterations of the \textbf{for} loop, while the number of bits in $\bin{a}$ is $m+1$, thus
larger by $1$. 
Consequently, in lines \ref{l:mbitS}-\ref{l:mbitE} the value of $\vr{v}$ is not flashed to $\vr{v}'$ nor restored back from $\vr{v}'$,
and hence $\vr{v}$ is forcedly $0$ at the end of $\mathcal{R}^*_{a, \true}$
and can be reused by the following commands.
Note that the \textbf{if} macro is used in line~\ref{l:dec2} as, due to the choice of $m$, the oldest bit $b_m$ of $\bin{N(n)}$ is $1$.

The above analysis applies equally well to every component $\mathcal{R}^*_{a, \false}$, 
as its computation is exactly the same as that of $\mathcal{R}^*_{a, \true}$,
except that the value of $\vr{u}$ is not changed.

\para{Dimension 7}
To estimate the dimension of the VASS represented by $\mathcal P$, notice that $\mathcal{I}'$ uses counters
$\vr{x}, \vr{x}',  \vr{y}, \vr{z}, \vr{e}, \vr{f}$ and components $\mathcal{R}^*_{s_i, p}$ use additionally $\vr{v}, \vr{v}',\vr{e}',\vr{u}$.
However, by Claim~\ref{claim:intro} the final values of $\vr{x},\vr{x}',\vr{z}$ computed by $\mathcal{I}'$ are $0$ in every
halting run of $\mathcal P$, hence
the three counters can be reused in components, 
which reduces the number of counters to~7.

\section{Doubly Exponential Shortest Runs}
\label{sec:fractions}


%
\begin{theorem}\label{thm:4vass}
%
There is a family of binary 4-VASS $(\mathcal{V}_n)_{n\in\N}$ of size $\O(n^3)$ such that every halting run of 
$\mathcal{V}_n$ is of length doubly exponential in $n$.
\end{theorem}
\newcommand{\dsize}{description size\xspace}
%
In this section we prove the theorem.
Define the \emph{\dsize} of an irreducible fraction $\frac{p}{q}$ as $\max\{p, q\}$.
We start with a key technical lemma stating existence of arbitrarily long increasing sequences of 
rationals greater than $1$, of \dsize exponential with respect to $k$,
with the property that the result of multiplying 
consecutive exponential powers of these rationals has only exponential
(and not doubly exponential) \dsize.

\begin{lemma}\label{lem:fractions}
For each $k \ge 1$ there are $k$ rational numbers 
\begin{align} \label{eq:increasing}
1 < f_1 < \ldots < f_k = 1 + \frac{1}{4^k},  
\end{align}
of \dsize bounded by $4^{k^2+k}$, such that the \dsize of $f$ defined by
\begin{align} \label{eq:magicfract}
f = (f_k)^{2^k} \cdot \ldots \cdot (f_2)^{2^2} \cdot  (f_1)^{2^1}
\end{align}
is bounded by $4^{2(k^2+k)}$.
\end{lemma}
%
%
%
A distinguished counter $\vr{x}$ in the 4-VASS $\mathcal{V}_k$ will play a special role:
in every halting run, $\vr{x}$ will be exactly multiplied by consecutive powers as in~\eqref{eq:magicfract}.
As the denominator of the irreducible form of $f_k$ is at least $2$, the counter $\vr{x}$, just before the very first multiplication
by $(f_k)^{2^k}$, must be divisible by the denominator of $f_k$ to the power $2^k$, which is
doubly exponential in~$k$.
In consequence, every halting run has to be doubly exponentially long.

\twocol{.35}{
\begin{algorithm}[H]
\caption{Program fragment $\hp(c,d)$; counters $\vr{x}, \vr{y}$ and $\vr{z}$ correspond 
to dimension 1, 2 and 3, respectively, of the VASS.}
\label{alg:hp}
\begin{algorithmic}[1]
\Loop \label{l:hpstart}
  \Loop\label{l:hp1}
    \State \dec{\vr{x}} \quad \inc{\vr{y}}
  \EndLoop  \label{l:hp1e}
  \Loop\label{l:hp2}
    \State \add{\vr{x}}{c} \quad\sub{\vr{y}}{d}
  \EndLoop\label{l:hp2e}
  \State \dec{\vr{z}}
\EndLoop \label{l:hpend}
\end{algorithmic}
\end{algorithm}
\movedown
}{\qquad}{.55}{
As before, the main difficulty is to turn weak multiplications into exact ones.
To this aim we will rely on Lemma~\ref{lem:fractions} 
and on a well-known weakly exponentiating 3-VASS gadget of Hopcroft and Pansiot~\cite{HopcroftP79}:
\movedown


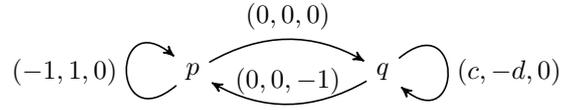
\begin{tikzpicture}[->,>=stealth',shorten >=1pt,auto,node distance=2.5cm,semithick]

\node (p) {$p$};
\node (q) [right of=p] {$q$};

\path[->]
(p) edge[->, in=140, out=220, min distance=0.5cm,loop] node {$(-1,1,0)$}(p)
(q) edge[->, in=320, out=40, min distance=0.5cm,loop] node {$(c,-d,0)$}(q)
(p) edge[->, bend left] node[above] {$(0,0,0)$} (q)
(q) edge[->, bend left] node[above] {$(0,0,-1)$} (p)
;
\end{tikzpicture}

The gadget is represented by the program fragment $\hp(c,d)$ shown in Algorithm~\ref{alg:hp}.
}
%

\moveup\moveup\moveup
\begin{proposition} \label{prop:hp}
Consider program fragment $\hp(c, d)$ for an irreducible fraction $\frac{c}{d}>1$, and initial values
$x_0, y_0, z_0$ of counters $\vr{x}$, $\vr{y}$ and $\vr{z}$.
In every run, the respective final values $x_1, y_1, z_1$ satisfy
\[x_1 + y_1 \leq (x_0 + y_0) \cdot \Big(\frac{c}{d}\Big)^{z_0 - z_1}. 
\]
Moreover, there is a run satisfying $x_1 = (x_0 + y_0)\cdot \Big(\frac{c}{d}\Big)^{z_0}$ if, and only if,
$x_0 + y_0$ is divisible by $d^{z_0}$. In this case $y_1 = z_1 = 0$.
\end{proposition}
\begin{proof}
The two inner loops (lines \ref{l:hp1}--\ref{l:hp2e}) coincide with the counter program fragment shown in 
Algorithm~\ref{alg:weak-multiplication}.
As the outer loop is executed $z_1 - z_0$ times,
the first part follows by Claim~\ref{claim:weak-multiplication}.

For the second part, assume $x_0 + y_0$ is divisible by $d^{z_0}$, and consider the unique run where all the loops are
iterated maximally, by which we mean:
\begin{itemize}
\item the outer loop (lines \ref{l:hpstart}--\ref{l:hpend}) is executed exactly $z_0$ times;
\item whenever execution of the first inner loop (lines \ref{l:hp1}--\ref{l:hp1e}) ends, the value of $\vr{x}$ is 0;
\item whenever execution of the second inner loop (lines \ref{l:hp2}--\ref{l:hp2e}) ends, the value of $\vr{y}$ is 0;
\end{itemize} 
Thus every execution of the two inner loops necessarily multiplies the sum $\vr{x} + \vr{y}$ by $\frac{c}{d}$, and
consequently, after $i$ iterations of the outer loop the values of respective counters $x', y', z'$ satisfy 
\begin{align} \label{eq:mid}
x' \ = \ (x_0 + y_0) \cdot \Big(\frac{c}{d}\Big)^{z_0 - i} \qquad
y' \ = \ 0 \qquad
z' \ = \ z_0 - i.
\end{align}
Repeating the multiplication $z_0$ times yields $x_1 = \Big(\frac{c}{d}\Big)^{z_0}$ and
$y_1 = z_1 = 0$, as required.

Conversely, suppose $x_1 = (x_0 + y_0) \cdot \Big(\frac{c}{d}\Big)^{z_0}$.
As $c$ and $d$ are co-primes, the sum of initial values $x_0 + y_0$ is thus forcedly divisible by $d^{z_0}$.
By the first part we know that the outer loop has been iterated maximally, hence $z_1 = 0$.
Then $y_1 = 0$ follows by the first part.
%
%
%
%
%
\end{proof}

\para{Construction of $\mathcal{V}_k$}
Fix $k\geq 1$. 
Let $f_i = \frac{a_i}{b_i}$, for $i\leq i\leq k$, be the fractions from Lemma~\ref{lem:fractions},
and let $f = \frac{a}{b}$ be the result of their multiplication as in~\eqref{eq:magicfract}.
We thus have:
\begin{align} \label{eq:ab}
\Big(\frac{a_1}{b_1}\Big)^2 \cdot
\Big(\frac{a_2}{b_2}\Big)^{2^2} \cdot
\ \ldots \ 
\cdot \Big(\frac{a_k}{b_k}\Big)^{2^k}
\quad = \quad \frac{a}{b}.
\end{align}
Algorithm~\ref{alg:4vass} (on the left below) shows the counter program representing the 4-VASS $V_k$ (on the right below),
using four counters $\vr{t}, \vr{x}, \vr{y}$ and $\vr{z}$. 
The constants appearing in increment and decrement commands are exponential in $k$, 
represented in binary in size $\mathcal{V}_k$ is $\O(k)$.
The length of $\mathcal{V}_k$ is $\O(k)$ and hence its size is $\O(k^3)$.




\twocol{.45}{
\begin{algorithm}[H]
\caption{Program representing 4-VASS $\mathcal{V}_k$ shown on the right.
Counters $\vr{t}, \vr{x}, \vr{y}$ and $\vr{z}$ correspond to consecutive dimensions.}
\label{alg:4vass}
\begin{algorithmic}[1]
\State \initialise
\State \inc{\vr{t}} \quad \add{\vr{x}}{1} \label{l:start}
\Loop
  \State \inc{\vr{t}} \quad \add{\vr{x}}{1}\label{l:firstloop}
\EndLoop
\For {\, $i$ \, := \, $k$ \, \textbf{down to } $1$}\label{l:firstfor}
  \State \add{\vr{z}}{2^i}
  \Loop \label{l:loopinfor}
    \Loop \label{l:for1start}
      \State \dec{\vr{x}} \quad \inc{\vr{y}}
    \EndLoop  \label{l:for1end}
    \Loop \label{l:for2start}
      \State \add{\vr{x}}{a_i} \quad\sub{\vr{y}}{b_i} \label{l:lastloopinfor}
    \EndLoop \label{l:for2end}
    \State \dec{\vr{z}} \label{l:lastfor}
  \EndLoop
\EndFor
\Loop\label{l:lastloop1}
  \State \sub{\vr{t}}{b} \quad \sub{\vr{x}}{a}\label{l:lastloop}
\EndLoop
\State \haltz{\vr{t}} 
\label{l:lasthalt}
\end{algorithmic}
\end{algorithm}
\movedown
}
{\qquad}{.55}{
\scalebox{.75}{

\begin{tikzpicture}[->,>=stealth',shorten >=1pt,auto,node distance=2.4cm,semithick]

\node(sp) {$\cdot$};
\node(s) [right of=sp] {$\cdot$};

\node (pk) [below = 1.6cm of sp] {$p_k$};
\node (qk) [right of=pk] {$q_k$};

\node (pk1) [below of=pk] {$p_{k-1}$};
\node (qk1) [right of=pk1] {$q_{k-1}$};

\node (dotst) [below = 0.5cm of pk1] {$\cdots$};
\node (dots) [below = 0.1cm of dotst] {};

\node (p1) [below = 1.2cm of dots] {$\ p_1\ $};
\node (q1) [right of=p1] {$\ q_1\ $};

\node (e) [below = 1.2cm of p1] {$\cdot$};

\path[->]
(s) edge[->] node[above] {\small $(1,1,0,0)$} (sp)
(sp) edge[->, in=50, out=130, loop] node {\small $(1,1,0,0)$}(sp)
(sp) edge[->] node[left] {\small $(0,0,0, 2^k)$} (pk)
(pk) edge[->, in=140, out=220, loop] node[left] {\small $(0,-1,1,0)$}(pk)
(qk) edge[->, in=320, out=40, loop] node[right] {\small $(0, a_k,-b_k,0)$}(qk)
(pk) edge[->, bend left=25]  (qk)
(qk) edge[->, bend left=25] node[below] {\small $(0,0,0,-1)$} (pk)
(pk) edge[->] node[left] {\small $(0,0,0, 2^{k-1})$} (pk1)
(pk1) edge[->, in=140, out=220, loop] node[left] {\small $(0,-1,1,0)$}(pk1)
(qk1) edge[->, in=320, out=40,loop] node[right=-14pt] {\small $(0, a_{k-1},-b_{k-1},0)$}(qk1)
(pk1) edge[->, bend left=25] (qk1)
(qk1) edge[->, bend left=25] node[below] {\small $(0,0,0,-1)$} (pk1)
(dots) edge[->] node[left, near start] {\small $(0,0,0, 2^{1})$} (p1)
(p1) edge[->, in=140, out=220,loop] node[left] {\small $(0,-1,1,0)$}(p1)
(q1) edge[->, in=320, out=40, loop] node[right] {\small $(0, a_{1},-b_{1},0)$}(q1)
(p1) edge[->, bend left=25]  (q1)
(q1) edge[->, bend left=25] node[below] {\small $(0,0,0,-1)$} (p1)
(p1) edge[->] node[left] {\small } (e)
(e) edge[->, in=230, out=310, loop] node[below] {\small $(-a,-b,0,0)$}(e)
;
\end{tikzpicture}

}
}

\begin{claim} \label{claim:hasrun}
For every $k\geq 0$, the 4-VASS $\mathcal{V}_k$ has a halting run.
\end{claim}
\begin{proof}
Put
$
N \ := \ \prod_{i = 1 \dots k} (b_i)^{2^i}.
$
By performing the first loop (lines 3--\ref{l:firstloop}) exactly $N-1$ times,
the run reaches the following valuation of counters $\vr{x}, \vr{y}, \vr{z}$:
\begin{align} \label{eq:jbase}
x_k \ = \ N \qquad\qquad y_k \ = \ z_k \ = \ 0.
\end{align}
Notice that the outer loop (lines \ref{l:loopinfor}--\ref{l:lastfor}) coincides with the program fragment $\hp(a_i, b_i)$.
We use the second part of Proposition~\ref{prop:hp} for consecutive iterations of the \textbf{for} macro.
The proposition allows us to derive a run 
where the values $x_j$, $y_j$, $z_j$ of counters $\vr{x}$, $\vr{y}$, $\vr{z}$,
after $k-j$ iterations of the \textbf{for} macro (for $j \in \{0, \ldots, k\}$), satisfy:
\begin{align} \label{eq:j}
x_j \ = \ N \cdot 
\Big(\frac{a_{j}}{b_{j}}\Big)^{2^{j}} \cdot
\ \ldots \ 
\cdot \Big(\frac{a_k}{b_k}\Big)^{2^k}
 \qquad y_j \ = \ z_j \ = \ 0.
\end{align}
Indeed, by induction with respect to $k-j$ (using~\eqref{eq:jbase} as induction base for $j = k$), we argue as follows:
if~\eqref{eq:j} holds then $x_j$ is divisible by $(b_{j-1})^{2^{j-1}}$, and hence by Proposition~\ref{prop:hp} there is a 
continuation of the run that yields
\[
x_{j-1} \ = \ x_j \cdot \Big(\frac{a_{j-1}}{b_{j-1}}\Big)^{2^{j-1}}
\qquad y_{j-1} \ = \ z_{j-1} \ = \ 0.
\]
In consequence, for $j = 0$ we obtain, using~\eqref{eq:ab}:
\[
x_0 \ = \ N \cdot \frac{a}{b} \qquad\qquad y_0 \ = \ z_0 \ = \ 0.
\]
As the counter $\vr{t}$ is not modified inside the \textbf{for} loop (lines~\ref{l:firstfor}--\ref{l:lastfor}), its value
is still equal to $N$ after \textbf{for} loop is finished.
Thus, by executing $N$ iterations of the last loop (in lines~\ref{l:lastloop1}--\ref{l:lastloop}) 
we reach the value $0$ of all the four counters $\vr{t}, \vr{x}, \vr{y}, \vr{z}$
and hence halt in line~\ref{l:lasthalt}. 
Summing up, every $\mathcal{V}_k$ admits a halting run.
\end{proof}

\begin{proof}[Proof of Theorem~\ref{thm:4vass}]
We argue that every halting run of $\mathcal{V}_k$ has length at least doubly exponential in $k$.
Consider an arbitrary halting run, i.e., a run reaching the final value $\vr{t} = 0$ 
in line~\eqref{l:lasthalt}.
As before, let $x_j$, $y_j$ and $z_j$, for $j = 0, \ldots, k$, stand for 
the values of counters $\vr{x}$, $\vr{y}$ and $\vr{z}$, respectively, after $k-j$ iterations of the \textbf{for} macro.
Let $x_k = N \geq 1$ be the value of the counters $\vr{t}$ and $\vr{x}$ after exiting from the first loop (lines~3--\ref{l:firstloop});
cf.~\eqref{eq:jbase}.
The counter $\vr{t}$ is not modified inside the \textbf{for} loop (lines~\ref{l:firstfor}--\ref{l:lastfor}). 
Thus the last loop (in lines~\ref{l:lastloop1}--\ref{l:lastloop}) 
has to be performed exactly $N$ times, which implies
\begin{align} \label{eq:geq}
x_0 \ \geq \ N \cdot \frac{a}{b}.
\end{align}
%

Let $n_k, n_{k-1}, \dots, n_1$ stand for the number of iterations of the outer loop 
(lines~\ref{l:loopinfor}--\ref{l:lastfor}) in consecutive iterations of  the \textbf{for} macro.
By the very structure of $\mathcal{V}_k$ we know that, for every $1 \leq i \leq k$,
\begin{align} \label{eq:nj}
\sum_{j = i}^k n_j \ \leq \ \sum_{j = i}^k 2^j.
\end{align}
We aim to show that the inequality~\eqref{eq:geq} implies $n_j = 2^j$ for every $j \in \{1, \ldots, k\}$.
As  the outer loop (lines \ref{l:loopinfor}--\ref{l:lastfor}) coincides with the program fragment $\hp(a_i, b_i)$, we may apply
the first part of Proposition~\ref{prop:hp} to derive, similarly as above:  
\begin{align} \label{eq:j2}
x_j \ \leq \ N \cdot 
\Big(\frac{a_{j}}{b_{j}}\Big)^{n_{j}} \cdot
\ \ldots \ 
\cdot \Big(\frac{a_k}{b_k}\Big)^{n_k}.
\end{align}
%
%
Claim~\ref{claim:flow} will imply that, roughly speaking, the biggest value of $x_j$ is obtained,
 if in every unfolding of the \textbf{for} macro we perform the maximal possible number of iterations
of the outer loop, and hence finish with the counter value $\vr{z} = 0$. 


\begin{claim}\label{claim:flow}
Assuming~\eqref{eq:nj}, 
$\Big(\frac{a_1}{b_1}\Big)^{n_1} \cdot
\Big(\frac{a_2}{b_2}\Big)^{n_2} \cdot
\ \ldots \ 
\cdot \Big(\frac{a_k}{b_k}\Big)^{n_k}
\ \leq \ \frac{a}{b}$.
The equality holds if, and only if, $n_j = 2^j$ for all $j \in \set{1,\ldots,k}$.
\end{claim}

\begin{proof}
For vectors $(n_1, \dots, n_k)$ satisfying~\eqref{eq:nj}, we define the function
$f(n_1, \dots, n_k) = \Big(\frac{a_1}{b_1}\Big)^{n_1} \cdot \ \ldots \ 
\cdot \Big(\frac{a_k}{b_k}\Big)^{n_k}$.
Thus~\eqref{eq:ab} says that
$
f(2^1, \dots, 2^k) = \frac{a}{b}.
$
Observe that any other vector $(n_1, \dots, n_k)$ satisfying~\eqref{eq:nj} is obtained from
$(2^1, \dots, 2^k)$ by applying a number of times one of the following two operations:
\begin{enumerate}
\item decrement some $n_i$ by $1$
\item decrement some $n_i$ by $1$ and increment $n_{i-1}$ by $1$.
\end{enumerate}
As any of this operations strictly decreases the value of $f$, Claim~\ref{claim:flow} follows.
\end{proof}

By the first part of Claim~\ref{claim:flow}, together with inequalities~\eqref{eq:nj} and~\eqref{eq:j2} we deduce $x_0 \ \leq \ N \cdot \frac{a}{b}$ which, combined with~\eqref{eq:geq} yields
the equality:
\begin{align*}
x_0 \ = \ N \cdot \frac{a}{b}.
\end{align*}
The latter equality, together with the second part of Claim~\ref{claim:flow}, implies $n_j = 2^j$ for all $j = 1 \dots k$.
As a consequence, the initial value $N$ of $\vr{x}$ is, due to the second part of Proposition~\ref{prop:hp}, divisible by $M = (b_k)^{2^k}$.
As $1 < \frac{a_k}{b_k} < 2$, we have $b_k \geq 2$, and hence $M$ is doubly exponential with respect to $k$.
It follows that the length of the run is also doubly exponential, as the first inner loop, 
in the first iteration of the \textbf{for} macro ($i = k$), is necessarily executed $N - 1 \geq M - 1$ times.
This concludes the proof of Theorem~\ref{thm:4vass}.
\end{proof}

\section{Conclusion}
\label{sec:conclusion}


Our three main results have provided non-trivial counter-examples that advance the state of the art in the challenging area of the complexity of the reachability problem for VASS (equivalently, VAS and Petri nets).  We have focussed on fixed dimension, and in particular, answered a central question that had remained open since \cite{BlondinFGHM15} and \cite{EnglertLT16}, namely whether reachability for flat VASS given in unary is decidable in nondeterministic logarithmic space for any fixed dimension, by establishing \np hardness in dimension~$7$.
Two specific matters that remain unresolved by this work are:
whether \np hardness of reachability for unary flat VASS is obtainable in any dimension less than~$7$ (and more than~$2$), and
whether binary VASS in dimension~$3$ can have doubly exponential shortest reachability witnesses.

We also remark that, although it has never been made precise, there seems to be an intriguing deep connection between the still open gap from \nl hardness to \np membership of reachability for unary flat $3$-VASS and the still open gap from \pspace hardness to \expspace membership of coverability for $1$-GVAS (1-VASS with pushdown)~\cite{LerouxST15,str17}.
Finally, we expect that the novel family of sequences of fractions developed in Section~\ref{sec:fractions} will have applications beyond the result obtained here.

\newpage

\bibliography{citat}

\begin{thebibliography}{10}

\bibitem{AngeliLS11}
David Angeli, Patrick~De Leenheer, and Eduardo~D. Sontag.
\newblock Persistence results for chemical reaction networks with
  time-dependent kinetics and no global conservation laws.
\newblock {\em {SIAM} Journal of Applied Mathematics}, 71(1):128--146, 2011.
\newblock \href {https://doi.org/10.1137/090779401}
  {\path{doi:10.1137/090779401}}.

\bibitem{BaldanCMS10}
Paolo Baldan, Nicoletta Cocco, Andrea Marin, and Marta Simeoni.
\newblock Petri nets for modelling metabolic pathways: a survey.
\newblock {\em Natural Computing}, 9(4):955--989, 2010.
\newblock \href {https://doi.org/10.1007/s11047-010-9180-6}
  {\path{doi:10.1007/s11047-010-9180-6}}.

\bibitem{BlondinFGHM15}
Michael Blondin, Alain Finkel, Stefan G{\"{o}}ller, Christoph Haase, and Pierre
  McKenzie.
\newblock Reachability in two-dimensional vector addition systems with states
  is {PSPACE}-complete.
\newblock In {\em {LICS}}, pages 32--43. {IEEE} Computer Society, 2015.
\newblock \href {https://doi.org/10.1109/LICS.2015.14}
  {\path{doi:10.1109/LICS.2015.14}}.

\bibitem{BojanczykDMSS11}
Miko{\l}aj Boja{\'{n}}czyk, Claire David, Anca Muscholl, Thomas Schwentick, and
  Luc Segoufin.
\newblock Two-variable logic on data words.
\newblock {\em {ACM} Trans. Comput. Log.}, 12(4):27:1--27:26, 2011.
\newblock \href {https://doi.org/10.1145/1970398.1970403}
  {\path{doi:10.1145/1970398.1970403}}.

\bibitem{BojanczykMSS09}
Miko{\l}aj Boja{\'{n}}czyk, Anca Muscholl, Thomas Schwentick, and Luc Segoufin.
\newblock Two-variable logic on data trees and {XML} reasoning.
\newblock {\em J. {ACM}}, 56(3):13:1--13:48, 2009.
\newblock \href {https://doi.org/10.1145/1516512.1516515}
  {\path{doi:10.1145/1516512.1516515}}.

\bibitem{BouajjaniE13}
Ahmed Bouajjani and Michael Emmi.
\newblock Analysis of recursively parallel programs.
\newblock {\em {ACM} Trans. Program. Lang. Syst.}, 35(3):10:1--10:49, 2013.
\newblock \href {https://doi.org/10.1145/2518188} {\path{doi:10.1145/2518188}}.

\bibitem{BurnsKY00}
Frank~P. Burns, Albert Koelmans, and Alexandre Yakovlev.
\newblock {WCET} analysis of superscalar processors using simulation with
  coloured {P}etri nets.
\newblock {\em Real-Time Systems}, 18(2/3):275--288, 2000.
\newblock \href {https://doi.org/10.1023/A:1008101416758}
  {\path{doi:10.1023/A:1008101416758}}.

\bibitem{ColcombetM14}
Thomas Colcombet and Amaldev Manuel.
\newblock Generalized data automata and fixpoint logic.
\newblock In {\em {FSTTCS}}, volume~29 of {\em LIPIcs}, pages 267--278. Schloss
  Dagstuhl, 2014.
\newblock \href {https://doi.org/10.4230/LIPIcs.FSTTCS.2014.267}
  {\path{doi:10.4230/LIPIcs.FSTTCS.2014.267}}.

\bibitem{ComonC00}
Hubert Comon and V{\'{e}}ronique Cortier.
\newblock Flatness is not a weakness.
\newblock In {\em {CSL}}, volume 1862 of {\em LNCS}, pages 262--276. Springer,
  2000.
\newblock \href {https://doi.org/10.1007/3-540-44622-2\_17}
  {\path{doi:10.1007/3-540-44622-2\_17}}.

\bibitem{Crespi-ReghizziM77}
Stefano Crespi{-}Reghizzi and Dino Mandrioli.
\newblock {P}etri nets and {S}zilard languages.
\newblock {\em Information and Control}, 33(2):177--192, 1977.
\newblock \href {https://doi.org/10.1016/S0019-9958(77)90558-7}
  {\path{doi:10.1016/S0019-9958(77)90558-7}}.

\bibitem{CzerwinskiLLLM19}
Wojciech Czerwi\'nski, S{\l}awomir Lasota, Ranko Lazi\'c, J{\'{e}}r{\^{o}}me
  Leroux, and Filip Mazowiecki.
\newblock The reachability problem for {P}etri nets is not elementary.
\newblock In {\em {STOC}}, pages 24--33. {ACM}, 2019.
\newblock \href {https://doi.org/10.1145/3313276.3316369}
  {\path{doi:10.1145/3313276.3316369}}.

\bibitem{DeckerHLT14}
Normann Decker, Peter Habermehl, Martin Leucker, and Daniel Thoma.
\newblock Ordered navigation on multi-attributed data words.
\newblock In {\em {CONCUR}}, volume 8704 of {\em LNCS}, pages 497--511.
  Springer, 2014.
\newblock \href {https://doi.org/10.1007/978-3-662-44584-6\_34}
  {\path{doi:10.1007/978-3-662-44584-6\_34}}.

\bibitem{DemriFP16}
St{\'{e}}phane Demri, Diego Figueira, and M.~Praveen.
\newblock Reasoning about data repetitions with counter systems.
\newblock {\em Logical Methods in Computer Science}, 12(3), 2016.
\newblock \href {https://doi.org/10.2168/LMCS-12(3:1)2016}
  {\path{doi:10.2168/LMCS-12(3:1)2016}}.

\bibitem{EnglertLT16}
Matthias Englert, Ranko Lazi\'c, and Patrick Totzke.
\newblock Reachability in two-dimensional unary vector addition systems with
  states is {NL}-complete.
\newblock In {\em {LICS}}, pages 477--484. {ACM}, 2016.
\newblock \href {https://doi.org/10.1145/2933575.2933577}
  {\path{doi:10.1145/2933575.2933577}}.

\bibitem{Esparza98}
Javier Esparza.
\newblock Decidability and complexity of {P}etri net problems --- an
  introduction.
\newblock In {\em Lectures on {P}etri Nets {I}}, volume 1491 of {\em LNCS},
  pages 374--428. Springer, 1998.
\newblock \href {https://doi.org/10.1007/3-540-65306-6\_20}
  {\path{doi:10.1007/3-540-65306-6\_20}}.

\bibitem{EsparzaGLM17}
Javier Esparza, Pierre Ganty, J{\'{e}}r{\^{o}}me Leroux, and Rupak Majumdar.
\newblock Verification of population protocols.
\newblock {\em Acta Inf.}, 54(2):191--215, 2017.
\newblock \href {https://doi.org/10.1007/s00236-016-0272-3}
  {\path{doi:10.1007/s00236-016-0272-3}}.

\bibitem{FribourgO97}
Laurent Fribourg and Hans Ols{\'{e}}n.
\newblock Proving safety properties of infinite state systems by compilation
  into {P}resburger arithmetic.
\newblock In {\em {CONCUR}}, volume 1243 of {\em LNCS}, pages 213--227.
  Springer, 1997.
\newblock \href {https://doi.org/10.1007/3-540-63141-0\_15}
  {\path{doi:10.1007/3-540-63141-0\_15}}.

\bibitem{GantyM12}
Pierre Ganty and Rupak Majumdar.
\newblock Algorithmic verification of asynchronous programs.
\newblock {\em {ACM} Trans. Program. Lang. Syst.}, 34(1):6:1--6:48, 2012.
\newblock \href {https://doi.org/10.1145/2160910.2160915}
  {\path{doi:10.1145/2160910.2160915}}.

\bibitem{GermanS92}
Steven~M. German and A.~Prasad Sistla.
\newblock Reasoning about systems with many processes.
\newblock {\em J. {ACM}}, 39(3):675--735, 1992.
\newblock \href {https://doi.org/10.1145/146637.146681}
  {\path{doi:10.1145/146637.146681}}.

\bibitem{Greibach78a}
Sheila~A. Greibach.
\newblock Remarks on blind and partially blind one-way multicounter machines.
\newblock {\em Theor. Comput. Sci.}, 7:311--324, 1978.
\newblock \href {https://doi.org/10.1016/0304-3975(78)90020-8}
  {\path{doi:10.1016/0304-3975(78)90020-8}}.

\bibitem{HaaseKOW09}
Christoph Haase, Stephan Kreutzer, Jo{\"{e}}l Ouaknine, and James Worrell.
\newblock Reachability in succinct and parametric one-counter automata.
\newblock In {\em {CONCUR}}, volume 5710 of {\em LNCS}, pages 369--383.
  Springer, 2009.
\newblock \href {https://doi.org/10.1007/978-3-642-04081-8\_25}
  {\path{doi:10.1007/978-3-642-04081-8\_25}}.

\bibitem{Hack74}
Michel Hack.
\newblock The recursive equivalence of the reachability problem and the
  liveness problem for {P}etri nets and vector addition systems.
\newblock In {\em {SWAT}}, pages 156--164. {IEEE} Computer Society, 1974.
\newblock \href {https://doi.org/10.1109/SWAT.1974.28}
  {\path{doi:10.1109/SWAT.1974.28}}.

\bibitem{HL18}
Piotr Hofman and S{\l}awomir Lasota.
\newblock Linear equations with ordered data.
\newblock In {\em {CONCUR}}, volume 118 of {\em LIPIcs}, pages 24:1--24:17.
  Schloss Dagstuhl, 2018.
\newblock \href {https://doi.org/10.4230/LIPIcs.CONCUR.2018.24}
  {\path{doi:10.4230/LIPIcs.CONCUR.2018.24}}.

\bibitem{HopcroftP79}
John~E. Hopcroft and Jean{-}Jacques Pansiot.
\newblock On the reachability problem for 5-dimensional vector addition
  systems.
\newblock {\em Theor. Comput. Sci.}, 8:135--159, 1979.
\newblock \href {https://doi.org/10.1016/0304-3975(79)90041-0}
  {\path{doi:10.1016/0304-3975(79)90041-0}}.

\bibitem{KKW14}
Alexander Kaiser, Daniel Kroening, and Thomas Wahl.
\newblock A widening approach to multithreaded program verification.
\newblock {\em {ACM} Trans. Program. Lang. Syst.}, 36(4):14:1--14:29, 2014.
\newblock \href {https://doi.org/10.1145/2629608} {\path{doi:10.1145/2629608}}.

\bibitem{Kanovich95}
Max~I. Kanovich.
\newblock {P}etri nets, {H}orn programs, linear logic and vector games.
\newblock {\em Ann. Pure Appl. Logic}, 75(1--2):107--135, 1995.
\newblock \href {https://doi.org/10.1016/0168-0072(94)00060-G}
  {\path{doi:10.1016/0168-0072(94)00060-G}}.

\bibitem{KarpM69}
Richard~M. Karp and Raymond~E. Miller.
\newblock Parallel program schemata.
\newblock {\em J. Comput. Syst. Sci.}, 3(2):147--195, 1969.
\newblock \href {https://doi.org/10.1016/S0022-0000(69)80011-5}
  {\path{doi:10.1016/S0022-0000(69)80011-5}}.

\bibitem{LerouxAG15}
H{\'{e}}l{\`{e}}ne Leroux, David Andreu, and Karen Godary{-}Dejean.
\newblock Handling exceptions in {P}etri net-based digital architecture: From
  formalism to implementation on {FPGA}s.
\newblock {\em {IEEE} Trans. Industrial Informatics}, 11(4):897--906, 2015.
\newblock \href {https://doi.org/10.1109/TII.2015.2435696}
  {\path{doi:10.1109/TII.2015.2435696}}.

\bibitem{Schmitz18cigar}
J{\'{e}}r{\^{o}}me Leroux and Sylvain Schmitz.
\newblock Reachability in vector addition systems is primitive-recursive in
  fixed dimension.
\newblock In {\em {LICS}}, pages 1--13. {IEEE}, 2019.
\newblock \href {https://doi.org/10.1109/LICS.2019.8785796}
  {\path{doi:10.1109/LICS.2019.8785796}}.

\bibitem{LerouxS04}
J{\'{e}}r{\^{o}}me Leroux and Gr{\'{e}}goire Sutre.
\newblock On flatness for 2-dimensional vector addition systems with states.
\newblock In {\em {CONCUR}}, volume 3170 of {\em LNCS}, pages 402--416.
  Springer, 2004.
\newblock \href {https://doi.org/10.1007/978-3-540-28644-8\_26}
  {\path{doi:10.1007/978-3-540-28644-8\_26}}.

\bibitem{LerouxST15}
J{\'{e}}r{\^{o}}me Leroux, Gr{\'{e}}goire Sutre, and Patrick Totzke.
\newblock On the coverability problem for pushdown vector addition systems in
  one dimension.
\newblock In {\em {ICALP}, Part {II}}, volume 9135 of {\em LNCS}, pages
  324--336. Springer, 2015.
\newblock \href {https://doi.org/10.1007/978-3-662-47666-6\_26}
  {\path{doi:10.1007/978-3-662-47666-6\_26}}.

\bibitem{LiDV17}
Yuliang Li, Alin Deutsch, and Victor Vianu.
\newblock {VERIFAS:} {A} practical verifier for artifact systems.
\newblock {\em {PVLDB}}, 11(3):283--296, 2017.
\newblock URL: \url{http://www.vldb.org/pvldb/vol11/p283-li.pdf}.

\bibitem{lipton76}
Richard~J. Lipton.
\newblock The reachability problem requires exponential space.
\newblock Technical Report~62, Yale University, 1976.
\newblock URL: \url{http://cpsc.yale.edu/sites/default/files/files/tr63.pdf}.

\bibitem{Mayr84}
Ernst~W. Mayr.
\newblock An algorithm for the general {P}etri net reachability problem.
\newblock {\em {SIAM} J. Comput.}, 13(3):441--460, 1984.
\newblock \href {https://doi.org/10.1137/0213029} {\path{doi:10.1137/0213029}}.

\bibitem{Meyer09}
Roland Meyer.
\newblock A theory of structural stationarity in the \emph{pi}-calculus.
\newblock {\em Acta Inf.}, 46(2):87--137, 2009.
\newblock \href {https://doi.org/10.1007/s00236-009-0091-x}
  {\path{doi:10.1007/s00236-009-0091-x}}.

\bibitem{PelegRA05}
Mor Peleg, Daniel~L. Rubin, and Russ~B. Altman.
\newblock Research paper: Using {P}etri net tools to study properties and
  dynamics of biological systems.
\newblock {\em {JAMIA}}, 12(2):181--199, 2005.
\newblock \href {https://doi.org/10.1197/jamia.M1637}
  {\path{doi:10.1197/jamia.M1637}}.

\bibitem{Petri62}
Carl~Adam Petri.
\newblock {\em Kommunikation mit Automaten}.
\newblock PhD thesis, Universität Hamburg, 1962.
\newblock URL:
  \url{http://edoc.sub.uni-hamburg.de/informatik/volltexte/2011/160/}.

\bibitem{Rackoff78}
Charles Rackoff.
\newblock The covering and boundedness problems for vector addition systems.
\newblock {\em Theor. Comput. Sci.}, 6:223--231, 1978.
\newblock \href {https://doi.org/10.1016/0304-3975(78)90036-1}
  {\path{doi:10.1016/0304-3975(78)90036-1}}.

\bibitem{RosierY86}
Louis~E. Rosier and Hsu{-}Chun Yen.
\newblock A multiparameter analysis of the boundedness problem for vector
  addition systems.
\newblock {\em J. Comput. Syst. Sci.}, 32(1):105--135, 1986.
\newblock \href {https://doi.org/10.1016/0022-0000(86)90006-1}
  {\path{doi:10.1016/0022-0000(86)90006-1}}.

\bibitem{Schmitz16siglog}
Sylvain Schmitz.
\newblock The complexity of reachability in vector addition systems.
\newblock {\em {SIGLOG} News}, 3(1):4--21, 2016.
\newblock \href {https://doi.org/10.1145/2893582.2893585}
  {\path{doi:10.1145/2893582.2893585}}.

\bibitem{str17}
Juliusz Straszy\'nski.
\newblock Complexity of the reachability problem for pushdown {P}etri nets.
\newblock Master's thesis, University of Warsaw, Faculty of Mathematics,
  Informatics, and Mechanics, 2017.
\newblock URL: \url{https://apd.uw.edu.pl/diplomas/155747}.

\bibitem{ValiantP75}
Leslie~G. Valiant and Mike Paterson.
\newblock Deterministic one-counter automata.
\newblock {\em J. Comput. Syst. Sci.}, 10(3):340--350, 1975.
\newblock \href {https://doi.org/10.1016/S0022-0000(75)80005-5}
  {\path{doi:10.1016/S0022-0000(75)80005-5}}.

\bibitem{Aalst15}
Wil M.~P. van~der Aalst.
\newblock Business process management as the ``killer app'' for {P}etri nets.
\newblock {\em Software and System Modeling}, 14(2):685--691, 2015.
\newblock \href {https://doi.org/10.1007/s10270-014-0424-2}
  {\path{doi:10.1007/s10270-014-0424-2}}.

\end{thebibliography}

\appendix


\section{Missing proofs}

\begin{proof}[Proof of Claim~\ref{claim:weak-multiplication}]
As $x' + y' = x_0 + y_0$ and $c > d$ we get:
\begin{align} \label{eq:cd}
x_1 + y_1 \le x' + \frac{c}{d} \cdot y' \leq \frac{c}{d} \cdot (x_0 + y_0).
\end{align}
We now concentrate on the second part of the claim. 
If $x' = y_1 = 0$ then $d \mid (x_0 + y_0)$ and thus $x_1 = (x_0 + y_0) \cdot \frac{c}{d}$.
For the opposite direction, if $y_1 \neq 0$ then $x_1 < x_1 + y_1 \le (x_0 + y_0) \cdot \frac{c}{d}$. 
If $x' \neq 0$ then by~\eqref{eq:cd} we get
\[
x_1 \le 
x' + \frac{c}{d} \cdot y' < \frac{c}{d} \cdot (x' + y') = \frac{c}{d}\cdot (x_0 + y_0). \qedhere
\]
\end{proof}


\begin{proof}[Proof of Claim~\ref{claim:mult}]
By Claim~\ref{claim:fraction} we get
$
x_1 + z_1 \le (x_{n+1} + z_{n+1}) \cdot \prod_{i = 1}^{n} \frac{i+1}{i} = (\vr{x}_{n+1} + \vr{z}_{n+1})\cdot (n+1).
$
Since $z_{n+1} = 0$ this implies the inequality.

Now we step to the second part of the claim. If $z_i = x'_i = 0$ for all $i = 1,\ldots,n$ then 
by Claim~\ref{claim:weak-multiplication} we get $x_i = x_{i+1} \cdot \frac{i+1}{i}$ for every $i$, which implies
$x_1 = x_{n+1} \cdot (n+1)$.

Conversely, suppose for some $i$ we have $z_i \neq 0$ or $x'_i \neq 0$. 
Then by Claim~\ref{claim:weak-multiplication} we get $x_i + z_i < (x_{i+1} + z_{i+1}) \cdot \frac{i+1}{i}$. 
Combined with Claim~\ref{claim:fraction} this yields $x_1 + z_1 < (x_{n+1} + z_{n + 1}) \cdot (n+1)$, 
which concludes the proof as $z_{n + 1} = 0$.
\end{proof}


\begin{proof}[Proof of Theorem~\ref{th:np}]
The size of $\mathcal P$ is polynomial in $n$, $k$ and $m$ and it can be computed in time polynomial with respect to the size of the input: $s_0$,
$S = \set{s_1, \ldots, s_k}$.
$\mathcal{P}$ represents a flat VASS since its explicite \textbf{goto} commands form a directed acyclic graph,
and \textbf{loop} macros are not nested. 
We prove that $\mathcal P$ has a halting run if, and only if, the instance 
$\set{s_0}, \set{s_1, \ldots, s_k}$ of the subset problem is positive.

\subparagraph*{($\Longleftarrow$)}
Fix a subset $R \subseteq S$ with $\sum_{s \in R} s = s_0$.
We define a halting run $\rho$ that starts (cf.~Claim~\ref{claim:intro}) by executing $\mathcal{I}'$ 
to compute $N(n)$ and $N(n) \cdot (k+1)$ in counters $\vr{e}$ and $\vr{f}$, respectively, and $0$ in the remaining counters $\vr{x}, \vr{y}$ and $\vr{z}$.
Then $\mathcal{R}^+_{s_0, \true}$ is executed, and finally
for every $1 \le i \le k$, if $i \in R$ then  $\rho$ jumps to $\mathcal{R}^-_{s_i,\true}$,
otherwise $\rho$ jumps to $\mathcal{R}^-_{s_i,\false}$. 
Inside every component $\mathcal{R}^*_{s_i,p}$ the run $\rho$ \emph{iterates all loops maximally}, by which we mean:
\begin{itemize}
\item the value of $\vr{v}$ is 0 at the exit of the loops in lines 3--6 and in lines 9--12;
\item the value of $\vr{v}'$ is 0 at the exit of the loop in lines 7--8;
\item the value of $\vr{e}'$ is 0 at the exit of the loop in lines 13--14.
\end{itemize}
It remains to observe that by iterating all loops maximally,
in every component $\mathcal{R}^*_{s_i,p}$, for $0 \leq i \leq k$, the counter $\vr{f}$ will be decremented by exactly $N(n)$,
and thus the value of $\vr{f}$ at the end of $\rho$ is zero.
Moreover,  $\mathcal{R}^+_{s_0,\true}$ sets the counter $\vr{u}$ to $s_0$, and for every $s_i \notin R$ the 
value of counter $\vr{u}$ 
is preserved by $\mathcal{R}^-_{s_i, \false}$, and for every $s_i \in R$ the counter $\vr{u}$ is decremented 
by $s_i$ in $\mathcal{R}^-_{s_i,\true}$.
Thus the value of the counter $\vr{u}$ is $0$ at the end of $\rho$, as well as the values of $\vr{y}$ and $\vr{f}$,
as required by the final \textbf{halt}.

\subparagraph*{($\Longrightarrow$)}
Consider a halting run $\rho$ of $\mathcal{P}$, and recall that after $\mathcal{I}'$ the counter 
$\vr{y}$ is not modified any more, and zero-tested by the final \textbf{halt} command of $\mathcal P$. 
By Claim~\ref{claim:intro} the values of $\vr{e}$ and $\vr{f}$ after
$\mathcal{I}'$ are $N(n)$ and $N(n) \cdot (k+1)$, respectively.


The sum of counters $\vr{e}$ and $\vr{e}'$ is invariantly equal $N(n)$ as decrement of one is always accompanied by
increment of the other.
Thus in every component $\mathcal{R}^*_{a,p}$ visited by $\rho$, 
the counter $\vr{f}$ is decreased by at most the initial value of $\vr{e}$, hence by at most $N(n)$. 
Finally, by construction of $\mathcal P$ the run $\rho$ passes through exactly $k+1$ components $\mathcal{R}^*_{a, p}$. 
Therefore, as $\vr{f}$ is zero-tested by the final \textbf{halt} command, we deduce.
\begin{claim}  \label{claim:fdecr}
The run $\rho$ decreases $\vr{f}$  by exactly $N(n)$ in every visited component $\mathcal{R}^*_{a,p}$.
\end{claim}
\noindent
In consequence, the initial values of component $\mathcal{R}^*_{s_i, p}$, for $0\leq i \leq k$, satisfy:
\begin{align*}
\vr{e}  = N(n) \qquad\qquad
\vr{f}  = N(n) \cdot (k+1 - i) \qquad\qquad
\vr{v}  = \vr{v}' = \vr{e}' = 0.
\end{align*}
Using Claim~\ref{claim:fdecr} we deduce.
\begin{claim}  \label{claim:itermax}
The run $\rho$ iterates all loops maximally in every visited component $\mathcal{R}^*_{a,p}$,
except possibly the last loop in line~\eqref{l:laste} in the last two components $\mathcal{R}^*_{s_k,p}$.
\end{claim}

\noindent
Possible non-maximal iteration of the last loop in $\mathcal{R}^*_{s_k,\true}$ and $\mathcal{R}^*_{s_k,\false}$ 
has no impact on the further analysis of the run $\rho$.
As a direct corollary we deduce:
\begin{claim}  \label{claim:udecr}
The run $\rho$ executes the command \incordec{\vr{u}} exactly $a$ times in every visited component $\mathcal{R}^*_{a,\true}$.
\end{claim}
\noindent
Therefore, the value of $\vr{u}$ is incremented by $s_0$ in component $\mathcal{R}^+_{s_0, \true}$.
Let $R \subseteq \set{s_1, \ldots, s_k}$ be the set of all $s_i$ such that $\rho$ passes through 
$\mathcal{R}^-_{s_i,\true}$. 
Again by Claim~\ref{claim:udecr}, for every $s_i \in R$ 
the value of $\vr{u}$ is decreased by $s_i$ in component $\mathcal{R}^-_{s_i, \true}$, and
for every $s_i \notin R$ the value of $\vr{u}$ is preserved in component $\mathcal{R}^-_{s_i, \false}$.
Since $\vr{u}$ is zero-tested by the final \textbf{halt} command, the instance of the \subsum problem 
is necessarily positive.
%
%
\end{proof}


\begin{proof}[Proof of Lemma~\ref{lem:fractions}]
For $1 \leq i \leq k$ put $r_i := \frac{4^k + 2^{k-i}}{4^k}$, and observe the following
(straightforward) equalities:
\[
\Big( \frac{1}{r_i} \Big)^{2^1} \cdot \Big( \frac{1}{r_i} \Big)^{2^2} \cdot \ldots \cdot \Big( \frac{1}{r_i} \Big)^{2^{i-1}} \cdot r_i^{2^i}
\ = \ r_i^2.
\]
Multiplying all these equalities yields the equality:  
\begin{align} \label{eq:magic}
f_1^{2^1} \cdot f_2^{2^2} \cdot \ldots \cdot f_k^{2^k} \ = \ f, \qquad
%
\text{ where } \quad
f_i \ = \  \frac{r_i}{r_{i+1} \cdot \ldots \cdot r_k}
\qquad\quad
f \ = \ \left(r_1 \cdot \ldots \cdot r_k\right)^2.
\end{align}
As numerators and denominators of all $r_i$ are bounded by $4^{k+1}$,
numerators and denominators of all $f_i$ are bounded by $4^{k^2 + k}$,
and numerator and denominator of $f$ are boun\-ded by $4^{2(k^2 + k)}$,
as required.

It remains to argue that the (in)equalities~\eqref{eq:increasing} hold. 
We notice the following relation between $r_i$ and $r_{i-1}$, for $1 < i \leq k$:
\begin{equation}\label{eq:square}
r_i^2 \ = \ \Big( 1 + \frac{2^{k-i}}{4^k} \Big)^2 \ > \ 1 + \frac{2^{k+1-i}}{4^k} \ = \ r_{i-1},
\end{equation}
which implies
\[
\frac{f_{i}}{f_{i-1}} \ = \ 
\frac{r_{i} \cdot (r_{i} \cdot \ldots \cdot r_k)}
{r_{i-1} \cdot  (r_{i+1} \cdot \ldots \cdot r_k)} 
\ = \ \frac{r_{i}^2}{r_{i-1}} \ > \ 1
\]
and hence $f_1 < f_2 < \ldots < f_k$. 
For $i = k$ we have $f_k = r_k = 1 + \frac{1}{4^k}$.  
It thus remains to show $f_1 > 1$,   
which is equivalent to
\begin{align} \label{eq:toprove}
r_1 > r_{2} \cdot \ldots \cdot r_k.
\end{align}
By~\eqref{eq:square} we deduce
$
r_k^{2^i}  >  r_{k-i},
$
by induction on $i$, which implies the following inequality:
\[
r_k^{2^{k-1} - 1} \ = \ r_k^{1 + 2 + 4 + \ldots + 2^{k-2}} \ > \ r_{2} \cdot \ldots \cdot r_k.
\]
For~\eqref{eq:toprove} it suffices to show,
relying on the above inequality, that
$r_1  >   r_k^{2^{k-1} - 1}.$
Put $N := 2^{k-1}-1$ for convenience. We thus need to prove:
\begin{align} \label{eq1}
r_1 \ > \  \Big(1 + \frac{1}{4^k}\Big)^N.
\end{align}
By inspecting the expansion of the right-hand side 
\[
\Big(1 + \frac{1}{4^k}\Big)^N \ = \ \sum_{i = 0}^N \ {N \choose i} \cdot \frac{1}{4^{ik}} 
\]
we observe that the right-hand side is bounded
by the sum of first $N$ elements of a geometric progression, which, in turn, is bounded by the sum of the whole
infinite one:
\[
\Big( 1 + \frac{1}{4^k}  \Big)^N \ \leq \ 1 + \frac{N}{4^k} + \frac{N^2}{4^{2k}}
+ \ldots + \frac{N^N}{4^{Nk}} \ < \ \frac{1}{1 - \frac{N}{4^k}}.
\] 
Thus for showing~\eqref{eq1} it is sufficient to prove the inequality
$
r_1  >  \frac{1}{1 - \frac{N}{4^k}},
$
which is equivalent to
\[
\Big( 1 - \frac{2^{k-1}-1}{4^k} \Big) \Big( 1 + \frac{2^{k-1}}{4^k} \Big) \ > \ 1.
\]
The latter inequality is easily verified to hold true as
\[
\frac{1}{4^k} \ > \  \frac{2^{k-i}-1}{4^k} \cdot  \frac{2^{k-i}}{4^k}.
\]
The inequality~\eqref{eq1} is proved, and hence so is Lemma~\ref{lem:fractions}.
\end{proof}


\end{document}